\documentclass[10pt,twocolumn,twoside]{IEEEtran}
\usepackage{cite}
\usepackage{amsmath,amssymb,amsfonts}
\usepackage{graphicx}
\usepackage{textcomp}
\usepackage[hidelinks]{hyperref}

\def\BibTeX{{\rm B\kern-.05em{\sc i\kern-.025em b}\kern-.08em
    T\kern-.1667em\lower.7ex\hbox{E}\kern-.125emX}}
\usepackage{xcolor}
\usepackage{amsthm}
\usepackage{algorithm} % Required for pseudo code
\usepackage{algpseudocode} % Required for pseudo code
\DeclareMathOperator*{\argmax}{arg\,max}

\usepackage{accents}

\usepackage{cleveref}

\allowdisplaybreaks

\newcommand\numberthis{\addtocounter{equation}{1}\tag{\theequation}}
\makeatletter
\def\BState{\State\hskip-\ALG@thistlm}
\makeatother
\begin{document}
\title{Safe Pricing Mechanisms for Distributed Resource Allocation with Bandit Feedback}
\author{Spencer Hutchinson, Berkay Turan, and Mahnoosh Alizadeh
\thanks{This work was supported by  NSF grant \#1847096.}
\thanks{S. Hutchinson, B. Turan and M. Alizadeh are with  Department of Electrical and Computer Engineering, University of California-Santa Barbara, Santa Barbara, CA 93106 USA (email: shutchinson@ucsb.edu; bturan@ucsb.edu; alizadeh@ucsb.edu). }}
\maketitle
\begin{abstract}
In societal-scale infrastructures, such as electric grids or transportation networks, pricing mechanisms are often used as a way to shape users' demand in order to lower operating costs and improve reliability.
Existing approaches to pricing design for safety-critical networks often require that users are queried beforehand to negotiate prices, which has proven to be challenging to implement in the real-world.
To offer a more practical alternative, we  develop learning-based pricing mechanisms that require no input from the users.
These pricing mechanisms aim to maximize the utility of the users' consumption by gradually estimating the users' price response over a span of $T$ time steps (e.g., days) while ensuring that the infrastructure network's safety constraints that limit the users' demand are satisfied at all time steps.
We propose two different algorithms for the two different scenarios when, 1) the utility function is chosen by the central coordinator to achieve a social objective and 2) the utility function is defined by the price response under the assumption that the users are self-interested agents. We prove that both algorithms enjoy $\tilde{\mathcal{O}} (T^{2/3})$ regret with high probability.
We then apply these algorithms to demand response pricing for the smart grid and  numerically demonstrate their effectiveness.
\end{abstract}

% \begin{IEEEkeywords}
% Safe Learning, Bandits, Optimization
% \end{IEEEkeywords}

% \theoremseparator{.}

\newtheorem{proposition}{Proposition}
\newtheorem{theorem}{Theorem}
\newtheorem{corollary}{Corollary}
\newtheorem{lemma}{Lemma}
\newtheorem{Fact}{Fact}
\newtheorem{remark}{Remark}
\newtheorem{assumption}{Assumption}
\newtheorem{definition}{Definition}

\newcommand{\eqdef}{\vcentcolon=}
\newcommand{\beq}{\begin{equation}}
\newcommand{\eeq}{\end{equation}}
\newcommand{\ie}{i.e., }
\newcommand\munderbar[1]{%
  \underaccent{\bar}{#1}}
  \newcommand{\bltxt}[1]{\textcolor{black}{#1}}
\newenvironment{bl}{\par\color{black}}{\par}

\section{Introduction}

In safety-critical infrastructure systems, such as power and transportation networks, prices or tolls are often used to improve efficiency while ensuring  safety constraints (e.g. power line or road capacities) are honored.
Optimal design of such prices requires knowledge of self-interested users' preferences/utility functions, which are not often apriori available to any central coordinator. As such,  one popular approach is to employ distributed resource allocation mechanisms such as network utility maximization (NUM), e.g., \cite{samadi2010optimal,mehr2017joint}. These approaches are well suited for finding optimal shadow prices in such multi-agent network systems via prescribed interactions between agents with private preferences \cite{nedic2009distributed,low1999optimization,palomar2007alternative}. After the distributed optimization protocol converges, optimal prices may be posted, and the users will adjust their demand in response to the posted prices.
However, in spite of their popularity in research papers, such  resource allocation mechanisms have not been widely implemented in real-world safety-critical networks, such as power systems, due to several factors including: 1) the need for back and forth communications with users  to negotiate over optimal prices; 2) they require fully automated personal demand management mechanisms to be adopted by each individual user in order to implement the distributed optimization protocol; 3) they require all users to take part in the  protocol and cooperate with the central entity.

To  circumvent these issues, in this paper, we adopt an alternative viewpoint wherein, instead of employing distributed mechanisms to find optimal prices given unknown user preferences, the central coordinator aims to \textit{learn} the users' preferences over a span of $T$ days through repeated interactions. Each day, the central coordinator posts a price and observes the users' noisy response through their resource consumption, and based on these observations, refines its knowledge of the users' preferences.

Adopting a learning-based pricing framework introduces several novel challenges that are not present in conventional approaches.
The first challenge is ensuring the infrastructure's \emph{safety constraints} when posting prices each day. As the price response of the users is being learnt and is hence not entirely known, the central coordinator needs to ensure that any posted price will not lead to user demand that will violate the network's constraints (e.g., power flow constraints in demand response applications).
The second challenge is ensuring the efficiency of the posted prices. Despite the fact that the central coordinator lacks full knowledge of the users' preferences, it still needs to ensure that the aggregate utility of the users due to the resource consumption is high over the span of $T$ days.

To model this problem, we study two different frameworks, \emph{Safe Price Response (SPR)} and \emph{Safe Utility Maximization (SUM)}.
The SPR and the SUM problems differ in that, in the former, the utility functions for different user groups are chosen by the central coordinator and, in the latter, the utility functions are defined by the price response function under the assumption that the users are self-interested agents.
SPR is appropriate when the optimization objective (i.e. the utility function) is a design choice of the central coordinator to achieve social objectives, while SUM is appropriate when the users are assumed to be self-interested agents and the goal is to maximize the total private benefit of these agents.

\begin{bl}
The contributions of this work are summarized as follows:
\begin{itemize}
    \item We introduce two new learning-based frameworks for pricing design in safety-critical networks that are applicable for the two typical settings where 1) the utility functions are designed by a central entity and 2) the utility functions are defined by the user's response to prices under the assumption that the users behave as self-interested agents.
    \item Relative to prior works such as \cite{nedic2009distributed,low1999optimization,palomar2007alternative}, our frameworks are more practical for safety-critical pricing applications because they \emph{do not} require prices to be negotiated with users beforehand to ensure safety;
    \item We propose two bandit algorithms for these frameworks and prove that they enjoy sublinear regret and satisfy the safety constraints at all rounds with high probability; 
    \item We apply these algorithms to demand response in the smart grid and demonstrate their effectiveness through simulation of a real distribution network.
\end{itemize}
\end{bl}

\emph{Related Work:}
Evidently, this work is related to existing approaches for demand management in safety-critical infrastructure.
Several works \cite{bitar2012deadline,han2015approximately,sun2018eliciting} take a mechanism design perspective, where the system users are modeled as strategic agents in their interactions with the central coordinator (e.g. a user may submit untruthful estimates of future demand to reduce their own costs).
In this paper, we take a different perspective in that we model the users' demand to be a private function of the price that does not change in response to the central coordinator's pricing policy.
This is more closely related to distributed resource allocation approaches, which are useful for finding optimal shadow prices in multi-agent systems with private utility functions.
The most relevant distributed resource allocation framework is Network Utility Maximization (NUM), which allows for a resource allocation problem with private utility functions to be decomposed such that it can be solved in a distributed fashion where a central coordinator communicates with each user \cite{palomar2006tutorial, chiang2007layering}. %, possibly in the form of pricing.
NUM has been applied to congestion control for internet networks \cite{low1999optimization, palomar2007alternative} as well as the control of power and transportation systems \cite{samadi2010optimal,li2011optimal,mehr2017joint}.
Recently, a NUM algorithm that respects stage-wise constraints was presented in \cite{turan2022safe}.
Another class of distributed resource allocation approaches consider a fully distributed system with cooperative agents and limited information sharing \cite{nedic2009distributed,nedic2014distributed}.
Our problem formulations are different than existing distributed resource allocation approaches in that there is noisy bandit feedback from users, a parametric form for the price response, and stage-wise safety constraints that must always be respected in spite of uncertainty about the users' response.

Given that this work is focused on learning in safety-critical applications, learning-based control techniques (surveyed in \cite{brunke2022safe}) are particularly relevant because they use previous data to improve performance while ensuring safe operation.
This includes learning-based adaptive control \cite{gahlawat2020l1, chowdhary2014bayesian}, learning-based robust control \cite{berkenkamp2015safe}, learning-based robust MPC \cite{koller2018learning} and model predictive safety certification \cite{wabersich2018linear}.
Although we use similar techniques to ensure safety, our problem fundamentally differs from the aforementioned approaches because the algorithm in our problem interacts with the same static environment at every time step, albeit with progressively more information (i.e. the environment does not evolve as a dynamical system).

In addition to learning-based control, there is also relevant literature on safe optimization.
This includes constrained optimization algorithms with unknown constraints and feasible iterates where the constraints are either linear \cite{usmanova2019safe} or nonlinear \cite{usmanova2020safe}, as well as online convex optimization with unknown constraints that need to be satisfied in the long term \cite{chen2018bandit} or constraints that need to be satisfied stage-wise
\cite{chaudhary2021safe}.
Several works have also considered the problem of safe learning under a Gaussian process prior \cite{sui2015safe, berkenkamp2021bayesian}.
However, none of these consider a multi-agent optimization problem with stage-wise constraints as we do here.

Most relevant to our work, prior work has also studied safety in linear stochastic bandits, where the reward is an unknown linear function of the action and the learner receives noisy bandit feedback of this action.
Different types of safety constraints have been considered, including constraints on the objective \cite{moradipari2020stage}, constraints on another linearly parameterized function with bandit feedback \cite{pacchiano2021stochastic} and constraints that are linear with respect to the decision variable and the unknown parameter \cite{amani2019linear}.
Although our algorithm and analysis are inspired by \cite{amani2019linear}, the key difference is that we have multiple constraints that jointly apply to multiple users (or equivalently, bandits) which necessitates different analysis techniques.
Refer to Section \ref{sec:reg_anal} for a more detailed comparison of the analysis.

This work studies a similar, but more general problem as the conference paper \cite{cdc2022paper}.
In particular, \cite{cdc2022paper} only considers the SPR problem formulation with $a_{ji} \geq 0$ for all $i\in[n],j\in[p]$, which allows for a simpler algorithm and analysis than what is presented here. The SUM formulation was not considered in \cite{cdc2022paper}.

\emph{Organization:}
Our study of the SPR and SUM problems are located in Section \ref{sec:price_resp} and \ref{sec:sum} respectively.
Additionally, the algorithms developed for these problems are applied to demand response in smart grid in Section \ref{sec:dr}.

\emph{Notation:}
For a positive integer $n$, we use $[n]$ to refer to the set of positive integers from $1$ to $n$ inclusive.
For a vector or matrix $A$, its transpose is denoted $A^{\top}$.
When $A$ is square, its minimum and maximum eigenvalues are denoted $\lambda_{min}(A)$ and $\lambda_{max} (A)$ respectively.
For a vector $v$ and positive definite matrix $P$, we use $\| v \|$ to refer to the euclidean norm of $v$ and $\| v \|_P $ to refer to $\sqrt{v^\top P v}$.
For a $d$-dimensional vector or $d$-tuple $v$ and positive integer $i$, we denote the $i$th element of $v$ as $v_i$.
In $d$ dimensions, the non-negative orthant and positive orthant are referred to as $\mathbb{R}_+^d$ and $\mathbb{R}_{++}^d$ respectively.
We use $\tilde{\mathcal{O}}$ to refer to Big-O notation that ignores logarithmic factors.
A vector of zeros and a vector of ones are indicated by $\mathbf{0}$ and $\mathbf{1}$ respectively, where the size is inferred by context.
For vectors $u$ and $v$, the notation $u \succ v$ indicates that each element of $u$ is greater than the corresponding element of $v$ and $u \succeq v$ indicates that each element of $u$ is greater than or equal to the corresponding element of $v$.
For a set $A$, $\mathbf{int}A$ refers to the interior of $A$ and $\mathbf{bd}A$ refers to the boundary of $A$.
\bltxt{The domain of a function $f$ is denoted by $\mathbf{dom} f$.}

\section{The Safe Price Response (SPR) Problem}
\label{sec:price_resp}

In this section, we first describe the SPR problem, and then present an algorithm and theoretical performance guarantees to address this problem.
The problem setup, algorithm, and regret analysis are presented in Sections \ref{sec:prob_set}, \ref{sec:alg}, and \ref{sec:reg_anal}.

\subsection{Problem Setup}
\label{sec:prob_set}

We pose a resource allocation problem involving repeated interactions between a central coordinator and $n$ users.
At each time step $t$ in horizon $[T]$, there is an interaction between each user $i\in[n]$ and the central coordinator in which the central coordinator chooses a price $\gamma_i^t$ and user $i$ responds with a resource consumption $x_i^t$.
The physical limits of the system are specified by $p$ linear constraints on the consumption vector $x^t = [x_1^t\ x_2^t\ ...\ x_n^t]^{\top}$.
The objective of the central coordinator is to maximize the total user utility (defined later), while ensuring that the constraints are satisfied at every $t\in~[T]$.

We adopt a parametric form for the price response function of the users. Specifically, we assume that the average resource consumption of user $i$ in response to the price $\gamma_i^t$ is given by its average price response function,
\begin{equation}
\label{eqn:price_resp}
    x_i^t = x_i(\gamma_i^t; \theta_i^*) = h_i(\gamma_i^t)^{\top} \theta_i^*,
\end{equation}
where $\theta_i^* \in \mathbb{R}^m_+$ is a nonzero parameter that is unknown to the central coordinator and $h_i: \mathbb{R} \rightarrow \mathbb{R}^m_{+}$ is a known continuous and non-increasing function \bltxt{where $\mathbf{dom} h_i = \mathbb{R}$}.
Equation \eqref{eqn:price_resp} models each user's price response as an unknown mix of given \textit{price response signatures}, where $h_i$ specifies the set of possible price response signatures that may be present in the user population and $\theta_i^*$ specifies what 
 (unknown) mix of these price response signatures make up user $i$'s price response.

\begin{bl}
For example, in the electricity demand response set up, the total price response of each user to posted prices is composed of the sum of usage of individual flexible appliances (e.g., electric vehicle or dishwasher) and flexible appliances have a limited number of ways to respond to prices (which is justified given the
automated nature of price response from home energy management systems, the limited types of flexible appliances, and
the common electricity load patterns of electricity customers). 
For example, time shiftable loads with similar energy demand and similar deadlines would show similar price response signatures. The response of an electric vehicle to posted prices can be determined by the charging rate, the amount of required charge, and the charging deadline. If these parameters were known, the response can be fully determined. However, since this is not the case,  we assume that each appliance can have one of a number of known price response signatures captured by $h_i(\gamma_i^t)$. The central coordinator does not know the exact combination of active price response signatures in each user's home (captured by $\theta_i^*$) and as such, needs to learn this information by choosing prices and observing the electricity usage of the homes.
We provide a more in-depth discussion of the electricity demand response example in Section \ref{sec:dr}.
\end{bl}
% \begin{bl}
% \begin{example}
% \label{ex:simple}
% To illustrate the practical implications of the price response model in \eqref{eqn:price_resp}, we consider a simple setting involving demand response in the smart grid
% (note that a more realistic, but complex, demand response setting is presented in Section \ref{sec:dr}).
% In each day of this setting, an aggregator (or similar entity) posts electricity prices for a set of homes and then the occupants of these homes will use various appliances according to the posted price of electricity and their individual preferences.
% Naturally, the usage of different appliances will have different sensitivities to price and therefore can be modeled with different price response signatures.
% For example, the usage of a dishwasher will be much less sensitive to price than the usage of a refrigerator, warranting a distinct price response signature for each of these two appliances.
% Although the aggregator can infer the price response signature of each appliance according to typical usage, there is no way for the aggregator to know the size (or quantity) of each appliance in each home.
% We can therefore model the price response of the users with \eqref{eqn:price_resp}, where each element of $h_i(\gamma_i)$ is the price response signature of each appliance as a function of price $\gamma_i$ and each element of $\theta_i$ is the size (or quantity) of each appliance in house $i$.
% \end{example}
% \end{bl}

The average price response function is non-increasing by definition, which is natural as consumption of a resource will generally not increase as price increases.
Also, due to inherent stochasticities present in the users' behaviors, the central coordinator observes the average resource consumption with some additive noise $\mu_i^t$. Specifically, on day $t$, the central coordinator observes the following response:
\begin{equation}
    \bar{x}_i^t = x_i(\gamma_i^t; \theta_i^*) + \mu_i^t.
\end{equation}
We assume the following noise model on $\mu_i^t$, which is often used in similar problems (e.g.  \cite{amani2019linear, chaudhary2021safe, abbasi2011improved}).
\begin{assumption}
\label{ass:noise_model}
    For all $i \in [n]$ and $t \in [T]$, the noise $\mu_i^t$ is conditionally $\sigma$-subgaussian such that, given the history $\mathcal{F}_i^t = \sigma(\gamma_i^1,\gamma_i^1,...,\gamma_i^{t+1},\mu_i^1, \mu_i^2, ...,\mu_i^t)$,  $\mathbb{E}[\mu_i^t | \mathcal{F}_i^{t-1} ] = 0$ and $\mathbb{E}[e^{\lambda \mu_i^t} | \mathcal{F}_i^{t-1} ] \leq \exp(\frac{\lambda^2 \sigma^2}{2}), \forall \lambda \in \mathbb{R}$.
\end{assumption}

In choosing the price vector $\gamma^t = [\gamma_1^t\ \gamma_2^t\ ...\ \gamma_n^t]^{\top}$ for the users at each time step $t \in [T]$, there are various objectives that the central coordinator might have depending on the specific application.
For example, in infrastructure systems that supply critical resources, it is important that the allocation of resources is fair such that under-served communities are not charged high prices, or large consumers do not block access to resources.
In any case, the central coordinator can design utility functions for each user to achieve the objective at hand.
Utility functions map the resource consumption of a user to utility and have been extensively studied (e.g. \cite{kelly1998rate, mo2000fair, lan2010axiomatic}).
Here, the utility function for user $i$ is the strictly increasing function $f_i : \mathbb{R} \rightarrow \mathbb{R}$, which means that the total utility for the system at time step $t$ is $\sum_{i=1}^n f_i (x_i^t)$.

Despite the high utility that comes with unrestricted consumption, there are physical limits on the system that restrict which consumption vectors are allowable.
These limits are specified by $p$ linear constraints on the users' consumptions, such that the set of feasible consumption vectors is compact and defined as
\begin{equation}
\label{eqn:feas_set}
    \bar{E} = \left\{ x \in \mathbb{R}^n : \sum_{i=1}^n a_{ji} x_i \leq c_j, \forall j \in [p] \right\},
\end{equation}
where $\{ a_{ji} \}_{i \in [n], j \in [p]}$ and $\{c_j\}_{j \in [p]}$ are known to the central coordinator.
Since the central coordinator only has access to noisy observations of the price response, it is in general impossible to design any method that enforces constraint \eqref{eqn:feas_set}  deterministically over the course of $T$ days without additional (but unrealistic) assumptions. 
As such, we take the next alternative, which is to slightly relax this requirement of safety and ensure it with a high probability jointly throughout the $T$ day operating time of our system.
That is, the central coordinator needs to ensure that every consumption vector $x^t$ is in $\bar{E}$ for all $t$ in $[T]$ with high probability.
Note that this is different from a regular chance constraint, which ensures constraint satisfaction with a certain probability per time step, meaning that the violation probability would compound as $T$ grows.
Since any feasible algorithm will ensure that all consumption vectors are in $\bar{E}$ with high probability, the following Lipschitz assumption on $f_i$ only needs to hold for feasible consumption vectors.
\begin{assumption}
\label{ass:lipschitz}
For all $i \in [n]$, the utility function $f_i$ is $M$-Lipschitz such that $|f_i(x_i^1) - f_i(x_i^2)| \leq M | x_i^1 - x_i^2 |$ for all $x^1$, $x^2$ in $\bar{E}$.
\end{assumption}

Given the model that has been specified so far, we can see that if the central coordinator had full information (i.e. knew $\{ \theta_i^* \}_{\forall i \in [n]}$) they would choose the price for every time step as
\begin{equation}
\label{eqn:opt_price}
    \gamma^* \in \argmax_{\gamma \in \bar{D}} \sum_{i=1}^n f_i (x_i (\gamma_i; \theta_i^*)),
\end{equation}
where $\gamma = [\gamma_1\ \gamma_2\ ...\ \gamma_n]^{\top}$ and
\begin{equation}
\label{eqn:feasp_set}
    \bar{D} = \left \{ \gamma \in \mathbb{R}^n : \sum_{i=1}^n a_{ji} x_i(\gamma_i; \theta_i^*) \leq c_j, \forall j \in [p] \right\},
\end{equation}
which we call the feasible price set.
The central coordinator cannot simply solve \eqref{eqn:opt_price} and choose $\gamma^*$ immediately because the $\theta_i^*$ are unknown to them.
Instead, the central coordinator uses the information from previous time steps (i.e. $\{ (\gamma^{\tau}, x^{\tau} ) \}_{\tau = 1}^{t-1}$) to choose the current price $\gamma_i^t$.
The central coordinator's performance in this task is measured by how close the total realized utility is to the optimal utility over $T$ time steps, which is referred to as regret:
\begin{equation}
R_T = \sum_{t=1}^T \sum_{i=1}^n \Big[ f_i \big( x_i ( \gamma_i^*; \theta_i^* ) \big) - f_i \big( x_i ( \gamma_i^t;\theta_i^* ) \big) \Big]
\end{equation}
The central coordinator's objective is to ensure that there is low regret and that, with high probability, every $\gamma^t$  is in $\bar{D}$ for all $t \in [T]$.

\begin{bl}
Given the stated problem, we make a technical assumption on the price response function in the following.
\end{bl}
\begin{assumption}
\label{ass:price_resp}
    For all $i \in [n]$, there exists a constant $S$ such that $\| \theta_i^* \| \leq S$.
    Also, $h_i$ satisfies $\lim_{\gamma_i \rightarrow -\infty} h_i(\gamma_i) = \boldsymbol{\infty}$ and $\lim_{\gamma_i \rightarrow \infty} h_i(\gamma_i) = \mathbf{0}$ for all $i \in [n]$.
\end{assumption}
\begin{bl}
The first part of this assumption ensures that the norm of $\theta_i^*$ is bounded by $S$ which is standard in the bandit learning literature, e.g. \cite{dani2008stochastic,abbasi2011improved}.
% Since $h_i$ is known to the central coordinator, they can choose an appropriate $L$ by finding the maximum value of $h_i$ over $\mathcal{Y}_i$.
In a real-world setting, an appropriate $S$ can be found with domain knowledge.
For example, in the demand response setup, an appropriate $S$ can be chosen by finding the worst-case estimates of the size of each appliance in each home.
% Note that our proposed algorithm (discussed in Section \ref{sec:alg}) will have sub-linear regret no matter how loose of a bound $S$ is.
% Nonetheless, $S$ should be chosen as tight as possible during implementation, as a smaller $S$ will reduce the final regret bound.
% However, $S$ can be chosen with worst-case estimates of the resource consumption in a particular setting.
The second part of Assumption \ref{ass:price_resp} ensures that there is a price (which may be negative) that will persuade the user to consume any  non-negative quantity of the resource.
This will generally be satisfied in real-world settings because the price could be negative, i.e. the central coordinator would pay the user to consume the resource. Such negative prices are occasionally used in power systems, for example. 
\end{bl}

As defined thus far, the problem does not provide the central coordinator with enough information to choose initial prices that satisfy the constraints.
To remedy this, we ensure by assumption that the prior knowledge on $\theta^*$, i.e. the fact that $\| \theta_i^* \| \leq S$ for all $i$ in $[n]$, is enough information for the central coordinator to construct a set of prices that is strictly within $\bar{D}$.
To state such an assumption, we first define the initial confidence set for $\theta^* = (\theta_1^*, \theta_2^*, ..., \theta_n^*)$ as $C^0 = C^0_1 \times C^0_2 \times ... \times C^0_n$, where
\begin{equation}
\label{eqn:init_set}
    C_i^0 = \{  \theta_i \in \mathbb{R}^m_+ : \| \theta_i \| \leq S\}
\end{equation}
for all $i$ in $[n]$.
Since $\theta^*$ is known to be in $C^0$, it follows that 
\begin{equation}
    D^0 := \left\{ \gamma \in \mathbb{R}^n : \sum_{i=1}^n a_{ji} \theta_i^T h_i(\gamma_i) \leq c_j - \zeta, \forall \theta \in C^0 \right\}
\end{equation}
is a subset of $\bar{D}$ for any $\zeta \geq 0$.
In Assumption \ref{ass:init_set}, we assume that $D^0$ is nonempty for some $\zeta$, providing the algorithm with a set of prices that are initially known to strictly satisfy the constraints.
We will also consider a set that is known to be larger than $\bar{D}$ given that $\theta^* \in C^0$,
\begin{equation}
    \hat{D}^0 := \left\{\gamma \in \mathbb{R}^n : \exists \theta \in C^0 \text{ s.t. } \sum_{i=1}^n a_{ji} \theta_i^T h_i(\gamma_i) \leq c_j \right\}.
\end{equation}
Since each price in $\hat{D}^0$ only needs to satisfy the constraints for \emph{some} $\theta \in C^0$, any algorithm that incorporates the information that $\theta_* \in C^0$ will only choose prices that are in $\hat{D}^0$.
Therefore, by assuming that the norm of $h_i$ is bounded for any price in $\hat{D}^0$, Assumption~\ref{ass:init_set} ensures that the norm of $h_i$ is bounded by a constant $L$ for any prices that are chosen by the algorithm.
In the real-world, such a constant $L$ can be simply calculated given that $h_i$ and $\hat{D}^0$ are known to the central coordinator.

% We are then ready to state Assumption \ref{ass:init_set}, which guarantees that there exists a nonempty set of prices, denoted by $D_0$, that will be contained in $\bar{D}$ if the true parameter is in $C_0$.
% Additionally, Assumption \ref{ass:init_set} provides a bound on the constraint parameters $a_{ji}$ and ensures the elements of $h_i$ are linearly independent for all prices in $D^0$. 
% This linear independence assumption specifies that the elements of $h_i$ are not scalar multiples of each other for all prices in $D^0$.
% This ensures that sampling $D^0$ will provide sufficient information about every dimension of $\theta_i^*$.
% In practice, this requires that the price response signatures are sufficiently different.

\begin{assumption}
\label{ass:init_set}
    There exists positive constants $\zeta$ and $\kappa$ such that the initial safe set $D^0$ is nonempty and $| a_{ji} | \leq \kappa$ for all $i$ in $[n]$ and $j$ in $[p]$.
    Additionally, there exists a positive constant $L$ such that $\max_{i \in [n]} \| h_i (\gamma_i) \| \leq L$ for all $\gamma \in \hat{D}^0$.
    Also, there does not exist a nonzero $\alpha_i \in \mathbb{R}^m$ for each $i \in [n]$ such that $\alpha_i ^{\top} h_i (\gamma_i) = 0$ for all $[\gamma_1\ \gamma_2\ ...\ \gamma_n]^{\top}$ in $D^0$.
\end{assumption}

\begin{bl}
In addition to what has already been discussed, Assumption \ref{ass:init_set} also specifies that 1) $| a_{ji} |$ is bounded by some constant $\kappa$ for all $i,j$ and that 2) the elements of $h_i$ are linearly independent on $D^0$.
Note that point 1 is mild given that the value of $| a_{ji} |$ is known and will be finite in any real-world application and therefore such a $\kappa$ can simply be calculated.
The linear independence assumption of point 2 specifies that the elements of $h_i$ are not scalar multiples of each other for all prices in $D^0$.
This ensures that sampling $D^0$ will provide sufficient information about every dimension of $\theta_i^*$.
In practice, this requires that the selected price response signatures are sufficiently different, which is a design choice.
\end{bl}

With the problem established, we develop an appropriate algorithm in the next section.

\subsection{Proposed Algorithm}
\label{sec:alg}
The proposed algorithm (Algorithm \ref{alg:safe_price}) first performs pure exploration by choosing prices in the initial safe set $D^0$ for an appropriately chosen duration $T'$, and then for the remaining time steps, chooses the prices via the optimism in the face of uncertainty (OFU) paradigm restricted to prices that are known to satisfy the constraints.
As proven in the analysis, this algorithm achieves sublinear regret while ensuring that the prices are in the feasible price set $\bar{D}$ for all time steps with high probability.

\begin{algorithm}[t]
    \caption{Safe Price Response Algorithm}
    \label{alg:safe_price}
    \begin{algorithmic}[1]
    \Require $\{ h_i \}_{i \in [n]}$, $\{ a_{ji} \}_{i \in [n], j \in [p]}$, $\{ c_j \}_{j \in [p]}$, $\{f_i\}_{i \in [n]}$, $S$, $L$
    \For{$t = 1$ to $T'$}
        \State Broadcast $\gamma^t \sim \text{Unif}(D^0)$ to the users.
        \State Observe noisy consumption $\bar{x}^t$.
    \EndFor
    \State Construct confidence set $C^{T'}$ with \eqref{eqn:conf_set}.
    \State Construct safe price set $D^{T'}$ with \eqref{eqn:safe_dec}.
    \For{$t = T' + 1$ to $T$}
        \State Find optimistic price $\gamma^t$ with \eqref{eqn:ofu_price}.
        \State Broadcast $\gamma^t$ to the users.
        \State Observe noisy consumption $\bar{x}^t$.
        \State Update confidence set $C^t$ with \eqref{eqn:conf_set}.
        \State Update safe price set $D^t$ with \eqref{eqn:safe_dec}.
    \EndFor
    \end{algorithmic}
\end{algorithm}

In order to both implement OFU and determine which prices are safe, the proposed algorithm uses previous price response information $\{ (\gamma^{\tau}, x^{\tau} ) \}_{\tau = 1}^{t-1}$ to construct confidence sets in which $\theta_i^*$ lie with high probability.
Given the regularized least-squares estimator for $\theta_i^*$ at time step $t$ with regularization paramater $\nu > 0$,
\begin{equation}
    \hat{\theta}_i^t = [V_i^t]^{-1} \sum_{s=1}^t h_i (\gamma_i^s) \bar{x}_i^s,
\end{equation}
where the gram matrix is
\begin{equation}
\label{eqn:gram_mat}
    V_i^t = \nu I + \sum_{s=1}^t h_i (\gamma_i^s) h_i (\gamma_i^s)^{\top},
\end{equation}
we use a modified version of the confidence set developed in \cite{abbasi2011improved}.
\begin{theorem}
\label{thm:conf_set}
(Theorem 1 in \cite{abbasi2011improved} modified for multiple users) Let Assumptions \ref{ass:noise_model} and \ref{ass:price_resp} hold. Recall the definition of $V_i^t$ in \eqref{eqn:gram_mat}. Then for all $i$ in $[n]$ and $t \geq 0$, we have with probability at least $1 - \delta$ that $\theta_i^*$ lies in the set
\begin{equation}
\label{eqn:conf_set}
C_i^t = \left\{ \theta_i \in \mathbb{R}^m_+ : \left\| \theta_i - \hat{\theta}_i^t \right\|_{V_i^t} \leq \sqrt{\beta^t}, \| \theta_i \| \leq S \right\}
\end{equation}
where
\begin{equation*}
    \sqrt{\beta^t} = \sigma \sqrt{m \log \Big( \frac{1 + t L^2 / \nu}{\delta/n} \Big) } + \sqrt{\nu} S.
\end{equation*}
\end{theorem}

The pure exploration phase of the algorithm is used to control the minimum eigenvalue of $V_i^{T'}$ and hence control the size of the confidence set $C^{T'}$.
In order to shrink the confidence set in a controlled manner, the algorithm samples prices IID from the initial safe set $D^0$.
Formally, we can state that as $\gamma^t \overset{\text{iid}}{\sim} \text{Unif}(D^0)$ for all $t$ in $[1,T']$.
As proven in Lemma \ref{lem:pos_def1}, this exploration strategy ensures that the parameter $\lambda_-$ (defined in \eqref{eqn:min_eig}) is strictly greater than zero.
This means that the confidence set $C^{T'}$ will shrink with $T'$ and therefore guarantees that the algorithm will have sublinear regret.

\begin{lemma}
\label{lem:pos_def1}
    Let Assumption \ref{ass:init_set} hold. Then, with Algorithm \ref{alg:safe_price} we have that 
    \begin{equation}
    \label{eqn:min_eig}
        \lambda_- := \min_{i \in [n]} \bigg[ \lambda_{min} \Big(\mathbb{E} \left[ h_i(\gamma_i^t) h_i(\gamma_i^t)^{\top} \right] \Big) \bigg] > 0,
    \end{equation}
    for all $t$ in $[1,T']$.
\end{lemma}
\begin{proof}
The proof is given in Lemma \ref{lem:pos_def} in Appendix~\ref{sec:apx_rt1}.
\end{proof}

For time steps after the pure exploration phase, the algorithm chooses actions optimistically within a conservative inner approximation of the feasible price set, which we call the safe price set.
% This safe price set can be specified with the confidence set for $\theta_i^*$ in \eqref{eqn:conf_set} by only including the prices that satisfy the constraints for any $\theta_i \in C_i^t$.
The safe price set is defined as
\begin{equation}
\label{eqn:safe_dec}
    D^t = \left\{ \gamma \in \mathbb{R}^n : \sum_{i=1}^n a_{ji} x_i(\gamma_i; \theta_i) \leq c_j, \forall j \in [p], \forall \theta \in C^t \right\}
\end{equation}
where $C^t = C_1^t \times C_2^t \times ... \times C_n^t$.
Equation \eqref{eqn:safe_dec} implies that for any $\gamma \in D^t$ and any $\theta \in C^t$, it holds that $[x_1(\gamma_1, \theta_1), ..., x_n(\gamma_n, \theta_n)]^\top$ is in the feasible consumption set $\Bar{E}$.
Since $\theta^*$ is in $C^t$ for all $t \in [T]$ with high probability (due to Theorem \ref{thm:conf_set}), any $\gamma \in D^t$ will yield a feasible consumption vector with the same probability.
Therefore, the algorithm ensures that the price vectors  at all time steps are feasible with high probability by choosing each price vector $\gamma^t$ from the safe price set $D^t$.
Among the price vectors in $D^t$, the algorithm chooses one that is optimistic, i.e. the algorithm finds a $\gamma^t$ such that
\begin{equation}
\label{eqn:ofu_price}
    (\gamma^t, \tilde{\theta}^t) \in \argmax_{(\gamma,\theta) \in D^{t-1} \times C^{t-1}} \sum_{i=1}^n f_i \left( h_i(\gamma_i)^{\top} \theta_i \right).
\end{equation}
For each time step after the pure exploration phase (when $t > T'$) the algorithm broadcasts the optimistic price found with \eqref{eqn:ofu_price}, observes the noisy consumption $\bar{x}^t$ and then updates the confidence set $C^t$.
In the next section we provide theoretical regret guarantees for the proposed algorithm.

\subsection{Regret Analysis}
\label{sec:reg_anal}

In this section, we prove that, with high probability, the regret of the proposed algorithm is $\tilde{\mathcal{O}}(T^{2/3})$ as given by Theorem \ref{thm:reg_bound}. 
This regret bound is comparable with similar safe learning algorithms as \cite{chaudhary2021safe} and \cite{amani2019linear} give the same order bound.

\begin{theorem}
\label{thm:reg_bound}
    Let Assumptions \labelcref{ass:init_set,ass:lipschitz,ass:noise_model,ass:price_resp} hold. Then with probability at least $1 - 2 \delta$, we have that the regret of Algorithm \ref{alg:safe_price} satisfies
    \begin{equation*}
    \label{eqn:reg_bound}
    \begin{split}
        R_T \leq & n M \max(LS,1) \sqrt{8 (T - T') \beta^T m \log \left( 1 + \frac{T L^2}{m \nu} \right)}\\
        & + 2 MnLST' + \frac{4 \sqrt{2} (T - T') \kappa M n^2 L^2 S \sqrt{\beta^T}}{\zeta \sqrt{2 \nu + \lambda_- T'}}
    \end{split}
    \end{equation*}
    when $T' \geq t_{\delta} = \frac{8 L^2}{\lambda_-} \log(\frac{nm}{\delta})$.
    In particular, choosing $T' = \max(n^{2/3} T^{2/3}, t_{\delta})$ ensures that $R_T \in \tilde{\mathcal{O}}(n^{5/3} T^{2/3})$.
\end{theorem}

The complete proof of Theorem \ref{thm:reg_bound} is given in Appendix \ref{sec:apx_thm2}.
This proof relies on a decomposition of the instantaneous regret that separates (I) the instantaneous regret due to the difference between the safe price set $D^t$ and the true price set $\bar{D}$, and (II) the instantaneous regret due to the size of the confidence set for $\theta_i^*$.
Given the definition of instantaneous regret,
\begin{equation}
    r_t =  \sum_{i=1}^n \left[ f_i \left( h_i(\gamma_i^*)^{\top} \theta_i^* \right) - f_i \left(h_i(\gamma_i^t)^{\top} \theta_i^* \right) \right],
\end{equation}
we have the decomposition
\begin{equation}
    r_t = r_t^I + r_t^{II},
\end{equation}
where
\begin{align}
\begin{split}
    r_t^I & = \sum_{i=1}^n \left[ f_i \Big( h_i(\gamma_i^*)^{\top} \theta_i^* \Big) - f_i \Big( h_i(\gamma_i^t)^{\top} \tilde{\theta}_i^t \Big) \right],\\
    r_t^{II} & =  \sum_{i=1}^n \left[ f_i \Big( h_i(\gamma_i^t)^{\top} \tilde{\theta}_i^t \Big) - f_i \Big( h_i(\gamma_i^t)^{\top} \theta_i^* \Big) \right].
\end{split}
\end{align}
Establishing the bound on $r_t^{II}$ uses similar techniques to the stochastic linear bandit analysis, such as \cite{abbasi2011improved}.
Bounding $r_t^I$ is somewhat more challenging and existing theory proves to be largely insufficient.

We bound $r_t^I$ for all time steps greater than $T'$ in Lemma~\ref{lem:term1}.

\begin{lemma}
\label{lem:term1} 
Let Assumptions \labelcref{ass:noise_model,ass:lipschitz,ass:price_resp,ass:init_set} hold.
Recall the definition of $\gamma_i^*$ and $\tilde{\theta}^t$ in \eqref{eqn:opt_price} and \eqref{eqn:ofu_price} respectively.
Then, the set of prices chosen by Algorithm 1 for time steps $t$ greater than $T' \geq t_{\delta}$, $\{\gamma_i^t\}_{\forall i \in [n], t > T'}$, satisfies
\begin{equation*}
\label{eqn:term1}
\begin{split}
r_t^I & = \sum_{i=1}^n \left[ f_i \Big( h_i(\gamma_i^*)^{\top} \theta_i^* \Big) - f_i \Big( h_i(\gamma_i^t)^{\top} \tilde{\theta}_i^t \Big) \right]\\
& \leq \frac{4 \sqrt{2} \kappa M n^2 L^2 S \sqrt{\beta^T}}{\zeta \sqrt{2 \lambda + \lambda_- T'}},
\end{split}
\end{equation*}
with probability at least $1 - 2 \delta$.
\end{lemma}

\emph{Proof sketch:} The complete proof of Lemma \ref{lem:term1} is given in Appendix \ref{sec:apx_rt1}. This proof draws inspiration from \cite{amani2019linear} in that it considers a line segment between a point in the initial safe set and the optimal solution, and then tracks the growth of the safe set along this line segment by relating it to the shrinkage of the paramater confidence set (i.e. $C^t$ in this case).
Despite the influence from \cite{amani2019linear}, our problem requires more complex work and new techniques to handle two primary challenges: (a) the fact that there are nonlinear basis functions (i.e. $h_i$), and (b) the fact that there multiple constraints that jointly apply to multiple users.

Due to challenge (a), we cannot take the natural approach of using a line segment in the price domain (i.e. a line segment from a point in $D^0$ to $\gamma^*$), because the constraint is nonlinear with respect to the price.
To work around this issue, we consider a line segment in the domain of $h = [h_1(\gamma_1)^{\top} \ h_2(\gamma_2)^{\top} \ ...\ h_n(\gamma_n)^{\top}]^{\top}$.
As a result, each of the constraints are linear with respect to any point on the line segment, making it feasible to bound.
However, this introduces additional challenges because there needs to be careful consideration of which values the $h$ vector can take given the range of each $h_i$ function.

Challenge (b) makes it difficult both to determine which constraint this line segment crosses and to bound the growth of the safe set across multiple users.
To address this challenge, we use the increasing property of $f_i$ to show that at least one constraint is tight on the optimal solution and use this fact to relate the growth of the safe set to the mininum eigenvalue of the gram matrix.

In the next section, we extend the work from the price response problem to a setting where the utility function is a property of the user rather than being chosen by the central coordinator.

\section{Safe Utility Maximization (SUM) Problem}
\label{sec:sum}

In this section, we consider the setting where the user utility functions $f_i(\cdot)$ are not designed by the central coordinator, but are instead defined by the price response of the users under the assumption that the users are self-interested agents.
In particular, the price response function $x_i(\gamma_i, \theta_i)$ now corresponds to the profit-maximizing consumption, with the \emph{profit} due to consumption $x_i$ taken to be the utility $f_i(x_i,\theta_i)$ minus the cost $\gamma_i x_i$.
That is, we want $f_i$ to be defined such that \begin{equation}\label{profitmax}
x_i(\gamma_i, \theta_i) = \argmax_{x_i} (f_i(x_i,\theta_i) - \gamma_i x_i).    
\end{equation}
This setting is especially useful because it captures the behavior of rational self-interested agents, which are prevalent in safety-critical infrastructure systems.
For example, electricity customers will choose an electricity consumption that maximizes the benefit (or utility) that they get from the electricity minus the costs of the electricity (e.g. \cite{li2011optimal,samadi2010optimal} use such a model). We will specify the specific structure of utility functions that can concurrently satisfy \eqref{profitmax} and \eqref{eqn:price_resp}.

Given that our problem is utility-maximizing and satisfies this profit-maximizing property, it can also be viewed as a safe version of the dual NUM problem (see \cite{palomar2006tutorial}) with a specific structure for the utility function and noisy observations of the consumption.
Our problem is considered to be \emph{safe} because, unlike conventional dual NUM, it ensures that the resource constraints are satisfied at each time step.
Therefore, our work may find further application in areas in which dual NUM algorithms have traditionally been used, as well as in safety-critical areas that may benefit from NUM-type algorithms.

The problem setup, proposed algorithm and regret analysis are given in Sections \ref{sec:sum_prob}, \ref{sec:sum_alg} and \ref{sec:sum_reg}, respectively.

\subsection{Problem Setup}
\label{sec:sum_prob}

% We adopt the same notation as SPR .
% Then, with the price response taken to be the profit-maximizing consumption in \eqref{profitmax}, we identify the class of utility functions that ensure the price response function has the parametric structure specified by \eqref{eqn:price_resp}. 
% We will see that this includes common utility function forms such as logarithm and reciprocal.

\bltxt{In this setting, the utility functions are not known to the central coordinator.
Instead, the utility functions, denoted by $\munderbar{f}_i$, are defined in terms of the price response function $x_i(\cdot)$ given that the price response is the profit-maximizing consumption.}
That is, the utility function for user $i$, denoted $\munderbar{f}_i: \mathbb{R}_{++} \times \mathbb{R}^m \rightarrow \mathbb{R}$, is differentiable with respect to the first argument and is implicitly defined as
\begin{equation}
\label{eqn:no_profit}
    x_i(\gamma_i, \theta_i) = \argmax_{x_i \in \mathbb{R}_{++}} \left( \munderbar{f}_i(x_i, \theta_i) - \gamma_i x_i \right),
\end{equation}
where $x_i(\cdot)$ is the price response function in \eqref{eqn:price_resp}.
It follows from its definition that $\munderbar{f}_i(x_i,\theta_i)$  represents the utility that a self-interested user gets from a consumption of $x_i$, given that her price response function is $x_i(\cdot,\theta_i)$.

To ensure that $\munderbar{f}_i$  \bltxt{is well defined and} satisfies the same properties of the utility functions as in the SPR setting (e.g. increasing, Lipschitz), we make the following modifications to the price response functions.
We first specify that $h_i$ is differentiable and strictly decreasing.
This is more restrictive than the SPR setting as $h_i$ is specified as continuous and non-increasing in that case.
\bltxt{These restrictions on $h_i$ ensure that $\munderbar{f}_i(\cdot,\theta_i)$ is unique up to an additive constant for a given price response function $x_i(\cdot,\theta_i)$ as proven in Appendix \ref{sec:apx_thm3}.}
We also make the following assumption on the price response, which is a stronger version of Assumption \ref{ass:price_resp} from the SPR setting.
\begin{assumption}
\label{ass:no_resp}
    (Replaces Assumption \ref{ass:price_resp})
    There exists positive constants $S$ and $\rho$ such that $\| \theta_i^* \| \leq S$ and $\mathbf{1}^{\top} \theta_i^* \geq \rho$.
    Also, the domain of $h_i$ is the positive reals where $\lim_{\gamma_i \rightarrow 0^+} h_i(\gamma_i) = \boldsymbol{\infty}$ and $\lim_{\gamma_i \rightarrow \infty} h_i(\gamma_i) = \mathbf{0}$.
\end{assumption}
The first part of Assumption \ref{ass:no_resp} is rather mild as it only adds the condition that $\mathbf{1}^{\top} \theta_i^* \geq \rho$ to Assumption \ref{ass:price_resp}.
However, the second part of Assumption \ref{ass:no_resp} is somewhat stronger than the equivalent part of Assumption \ref{ass:price_resp} because it ensures that there exists a \emph{positive} price (versus a real-valued price in Assumption \ref{ass:price_resp}) that will compel the user to consume any non-negative quantity of the resource.
\bltxt{Examples of basis functions that satisfy these assumptions (where $m=1$, i.e. $\theta_i$ is a scalar) are $h_i(\gamma_i) = 1/ \gamma_i$ which corresponds to $x_i(\gamma_i,\theta_i) = \theta_i /\gamma_i$ and $\munderbar{f}_i(x_i,\theta_i) = \theta_i \log(x_i)$, as well as $h_i(\gamma_i) = 1/ \sqrt{\gamma_i}$ which corresponds to $x_i(\gamma_i,\theta_i) = \theta_i /\sqrt{\gamma_i}$ and $\munderbar{f}_i(x_i,\theta_i) = - \theta_i^2 /x_i$.~\footnote{\bltxt{We give examples with $m=1$ because in more complicated settings, there may not be a closed-form expression for the utility function. See \eqref{eqn:no_util} for an integral expression for the utility function.}}}

Now that the utility function for this setting has been specified, we define the optimal price vector for this setting
\begin{equation}
    \munderbar{\gamma}^* \in \text{arg}\max_{\gamma \in \bar{D}} \sum_{i=1}^n \munderbar{f}_i ( x_i(\gamma_i, \theta_i^*) , \theta_i^*),
\end{equation}
where $\bar{D}$ is the feasible price set defined in \eqref{eqn:feasp_set}.
Using the definition of optimal prices, we can then define the regret due to the prices $\{ \munderbar{\gamma}_i^t \}_{i \in [n], t \in [T]}$ as
\begin{equation}
\munderbar{R}_T = \sum_{t=1}^T \sum_{i=1}^n \left[ \munderbar{f}_i (x_i (\munderbar{\gamma}_i^*, \theta_i^*),\theta_i^*) - \munderbar{f}_i (x_i (\munderbar{\gamma}_i^t,\theta_i^*),\theta_i^*) \right].
\end{equation}
Note that the only difference between the definition of both the optimal price and regret in this setting versus the SPR setting is that $\munderbar{f}_i$ is used in place of the SPR utility function $f_i$.

Given that Assumption \ref{ass:no_resp} provides an additional restriction on the parameter $\theta_i^*$ (i.e. the condition that $\mathbf{1}^{\top} \theta_i^* \geq \rho$), we need to define the initial confidence set for this setting (equivalent to \eqref{eqn:init_set}),
\begin{equation}
    \munderbar{C}_i^0 = \{  \theta_i \in \mathbb{R}^m_+ : \| \theta_i \| \leq S, \mathbf{1}^{\top} \theta_i \geq \rho \}
\end{equation}
with $\munderbar{C}^0 = \munderbar{C}^0_1 \times \munderbar{C}^0_2 \times ... \times \munderbar{C}^0_n$.
As in the SPR setting, we use $\munderbar{C}^0$ to define a set of prices that is contained in $\bar{D}$ and a set of prices that contains $\bar{D}$ which are $\munderbar{D}^0 = \{\gamma \in \mathbb{R}^n : \sum_{i=1}^n a_{ji} \theta_i^T h_i(\gamma_i) \leq c_j - \zeta, \forall \theta \in \munderbar{C}^0 \}$ and $\munderbar{\hat{D}}^0 := \left\{\gamma \in \mathbb{R}^n : \exists \theta \in \munderbar{C}^0 \text{ s.t. } \sum_{i=1}^n a_{ji} \theta_i^T h_i(\gamma_i) \leq c_j  \right\}$, respectively.
In the following assumption, we assume that $\munderbar{D}^0$ is nonempty, $h_i(\gamma_i)$ has bounded norm for all $\gamma \in \munderbar{\hat{D}}^0$ and that the elements of $h_i$ are linearly independent (equivalent to Assumption \ref{ass:init_set}).
\begin{assumption} 
\label{ass:no_initset}
    (Replaces Assumption \ref{ass:init_set})
    There exists positive constants $\zeta$ and $\kappa$ such that $\munderbar{D}^0$ is nonempty and $| a_{ji} | \leq \kappa$ for all $i$ in $[n]$ and $j$ in $[p]$.
    Additionally, there exists a positive constant $L$ such that $\max_{i \in [n]} \| h_i (\gamma_i) \| \leq L$ for all $\gamma \in \munderbar{\hat{D}}^0$
    Also, there does not exist a nonzero $\alpha_i \in \mathbb{R}^m$ for each $i \in [n]$ such that $\alpha_i ^{\top} h_i (\gamma_i) = 0$ for all $[\gamma_1\ \gamma_2\ ...\ \gamma_n]^{\top}$ in $\munderbar{D}^0$.
\end{assumption}
Note that Assumption \ref{ass:no_initset} is the same as Assumption \ref{ass:init_set} from the SPR setting, except it incorporates the additional prior information that $\mathbf{1}^{\top} \theta_i \geq \rho$ for all $i$ in $[n]$.

We then use $\munderbar{C}_i^0$ to state an assumption which ensures that $\munderbar{f}_i$ is Lipschitz.
This assumption uses the inverse of the price response function, which we denote $g_i (x_i, \theta_i)$ such that $x_i(g_i (x_i, \theta_i),\theta_i) = x_i$ for any $x_i$ in $\mathbb{R}_{++}$ and $g_i(x_i (\gamma_i, \theta_i),\theta_i) = \gamma_i$ for any $\gamma_i$ in $\mathbb{R}_{++}$.
In Lemma \ref{lem:no_equiv} in Appendix \ref{sec:apx_no_pre}, $g_i (x_i, \theta_i)$ is proven to exist and to be equal to $\frac{\partial}{\partial x_i}\munderbar{f}_i (x_i, \theta_i)$ for all $x_i$ in $\mathbb{R}_{++}$.

\begin{assumption} \label{ass:no_bound}
For all $x$ in $\bar{E}$ and $\theta$ in $\munderbar{C}^0$, there exists positive constants $\Gamma$, $L$ and $K$ such that $g_i (x_i, \theta_i) \leq \Gamma$,  $\| h_i( g_i (x_i, \theta_i) ) \| \leq L$, and $h_i'(g_i (x_i, \theta_i)) \preceq -\mathbf{1} K$ for all $i \in [n]$.
Additionally, there exists a point $x^0$ in $E$ such that $\munderbar{f}_i(x_i^0, \cdot)$ is $\eta$-Lipschitz on $\munderbar{C}_i^0$ for all $i$ in $[n]$.
\end{assumption}
The first part of Assumption \ref{ass:no_bound} provides bounds on $g_i (x_i, \theta_i)$, $h_i( g_i (x_i, \theta_i) )$ and $h_i'( g_i (x_i, \theta_i) )$ for values of $x_i$ and $\theta_i$ that the central coordinator might use as arguments for the utility function when estimating the optimal utility (i.e. the central coordinator initially knows that $\theta^*$ is in $\munderbar{C}^0$ and that the optimal $x$ is in $\bar{E}$).
The second part of Assumption \ref{ass:no_bound} is mild as it only states that $\munderbar{f}_i(x_i,\cdot)$ is Lipschitz for a single value of $x_i = x_i^0$.
Along with the other assumptions, this is sufficient to ensure that $\munderbar{f}_i(x_i,\cdot)$ is Lipschitz for all $x \in \bar{E}$.
In order to bound $x^0$ in the analysis, we use the fact that $\bar{E}$ is bounded by definition (in Section \ref{sec:prob_set}) to define the positive constant $\xi$ as satisfying $|x^1 - x^2| \preceq \mathbf{1} \xi$ for all $x^1$ and $x^2$ in $\bar{E}$.
In the next section, we propose an algorithm to address this problem setup.

\subsection{Proposed Algorithm}
\label{sec:sum_alg}

\begin{algorithm}[t]
    \caption{Safe Utility Maximization Algorithm}
    \label{alg:net_opt}
    \begin{algorithmic}[1]
    \Require $\{ h_i \}_{i \in [n]}$, $\{ a_{ji} \}_{i \in [n], j \in [p]}$, $\{ c_j \}_{j \in [p]}$, $\{\munderbar{f}_i\}_{i \in [n]}$, $S$, $L$, $\rho$, $x^0$, $f^0$
    \For{$t = 1$ to $T'$}
        \State Broadcast $\gamma^t \sim \text{Unif}(\munderbar{D}^0)$ to the users.
        \State Observe noisy consumption $\bar{x}^t$.
    \EndFor
    \State Construct confidence set $\munderbar{C}^{T'}$ with \eqref{eqn:no_conf_set}.
    \State Construct safe price set $\munderbar{D}^{T'}$ with \eqref{eqn:no_safe_dec}.
    \For{$t = T' + 1$ to $T$}
        \State Choose some $\munderbar{\check{\theta}}^t$ in $\munderbar{C}^{t-1}$.
        \State Find optimistic price $\munderbar{\gamma}^t$ with \eqref{eqn:no_ofu_price}.
        \State Broadcast $\munderbar{\gamma}^t$ to the users.
        \State Observe noisy consumption $\bar{x}^t$.
        \State Update confidence set $\munderbar{C}^t$ with \eqref{eqn:no_conf_set}.
        \State Update safe price set $\munderbar{D}^t$ with \eqref{eqn:no_safe_dec}.
    \EndFor
    \end{algorithmic}
\end{algorithm}

The proposed algorithm for this setting (Algorithm \ref{alg:net_opt}) operates nearly the same as the SPR algorithm with the exception being that the confidence set and the optimistic price are defined differently.
Incorporating Assumption \ref{ass:no_resp} into \eqref{eqn:conf_set}, the confidence set for $\theta_i^*$ in this setting is
\begin{equation}
\label{eqn:no_conf_set}
\begin{split}
    \munderbar{C}_i^t = \bigg\{ & \theta_i \in \mathbb{R}^m_+ : \\
    & \left\| \theta_i - \hat{\theta}_i^t \right\|_{V_i^t} \leq \sqrt{\beta^t}, \| \theta_i \| \leq S, \mathbf{1}^{\top} \theta_i \geq \rho \bigg\}.
\end{split}
\end{equation}
The safe price set is then
\begin{equation}
\label{eqn:no_safe_dec}
    \munderbar{D}^t = \left\{ \gamma \in \mathbb{R}^n : \sum_{i=1}^n a_{ji} x_i(\gamma_i; \theta_i) \leq c_j, \forall j \in [p], \forall \theta \in \munderbar{C}^t \right\}
\end{equation}
where $\munderbar{C}^t = \munderbar{C}_1^t \times \munderbar{C}_2^t \times ... \times \munderbar{C}_n^t$.
Using these definitions, the optimistic price is found by first choosing some $\munderbar{\check{\theta}}^t$ in $\munderbar{C}^{t-1}$ and then solving
\begin{equation}
\label{eqn:no_ofu_price}
    (\munderbar{\gamma}^t, \munderbar{\tilde{\theta}}^t) \in \argmax_{(\gamma,\theta) \in \munderbar{D}^{t-1} \times \munderbar{C}^{t-1}} \sum_{i=1}^n \munderbar{f}_i \left( h_i(\gamma_i)^{\top} \theta_i, \munderbar{\check{\theta}}_i^t \right)
\end{equation}
Note that this price update is nearly the same as the price update in the SPR setting, with the exception being that $\munderbar{f}_i(\cdot, \munderbar{\check{\theta}}^t_i)$ is used instead of $f_i(\cdot)$ where $\munderbar{\check{\theta}}^t$ is chosen arbitrarily in $\munderbar{C}^{t-1}$.
Although the optimistic price update \eqref{eqn:no_ofu_price} maximizes the approximate utility $\sum_{i =1}^n \munderbar{f}_i(\cdot, \munderbar{\check{\theta}}^t_i)$ rather than the true utility $\sum_{i =1}^n \munderbar{f}_i(\cdot, \theta_i^*)$ (since $\theta^*$ is unknown), the algorithm still enjoys sublinear regret because the difference between $\munderbar{f}_i(x_i, \munderbar{\check{\theta}}^t_i)$ and $\munderbar{f}_i(x_i, \theta_i^*)$ shrinks with time horizon.
This is due to the pure exploration phase, which ensures that $C^{t-1}$ shrinks with the time horizon and therefore that the distance between $\munderbar{\check{\theta}}^t$ and $\theta^*$ shrinks as well.
The complete regret bound for this algorithm is given in the next section.

\subsection{Regret Analysis}
\label{sec:sum_reg}

In this section, we extend the SPR regret analysis in Section \ref{sec:reg_anal} to this setting.
The main result is Theorem \ref{thm:no_reg}, which gives the regret bound for this setting.

\begin{theorem}
\label{thm:no_reg}
Let Assumptions \labelcref{ass:noise_model,ass:no_resp,ass:no_initset,ass:no_bound} hold.
Then, with probability at least $1 - 2 \delta$, we have that
\begin{equation*}
    \label{eqn:no_reg_bound}
    \begin{split}
        \munderbar{R}_T \leq & n \Gamma \max(LS,1) \sqrt{8 (T - T') \beta^T m \log \left( 1 + \frac{T L^2}{m \nu} \right)}\\
        & + 2 \Gamma nLST' + \frac{4 \sqrt{2} (T - T') \kappa \Gamma n^2 L^2 S \sqrt{\beta^T}}{\zeta \sqrt{2 \nu + \lambda_- T'}}\\
        & + \frac{4 \sqrt{2} n (T - T') ( \eta + \frac{\xi L}{\rho K}) \sqrt{\beta^T}}{ \sqrt{2 \nu + \lambda_- T'}}.
    \end{split}
\end{equation*}
    when $T' \geq t_{\delta} = \frac{8 L^2}{\lambda_-} \log(\frac{nm}{\delta})$.
    In particular, choosing $T' = \max(n^{2/3} T^{2/3}, t_{\delta})$ guarantees that $\munderbar{R}_T \in \tilde{\mathcal{O}}(T^{2/3}n^{5/3})$.
\end{theorem}

The proof of Theorem \ref{thm:no_reg} is in Appendix \ref{sec:apx_thm3}. 
We can see that the first three terms of the regret bound in Theorem \ref{thm:no_reg} match the bound for the SPR setting (Theorem \ref{thm:reg_bound}) except that $\Gamma$ appears in this bound where $M$ appears in the SPR bound.
The fourth term in the regret bound comes from the error in the second argument of the utility function, i.e. the difference between $\munderbar{f}_i(x_i, \check{\munderbar{\theta}}_i^t)$ and $\munderbar{f}_i(x_i,\theta_i^*)$ for some $x_i$.
However, the bound in Theorem \ref{thm:no_reg} is still the same order as Theorem \ref{thm:reg_bound} from the SPR setting.

\section{Application to Demand Response in Smart Grid}
\label{sec:dr}

In this section, we apply the SPR and SUM algorithms to demand response (DR) in smart grid.
DR is a mechanism by which an aggregator (or other organization that supplies power) can modify the electricity usage of its customers, sometimes through variable pricing.
This is advantageous because it can reduce the costs for the aggregator and its customers, and improve reliability \cite{qdr2006benefits}.
One popular type of DR program is day-ahead real-time pricing (RTP), where each day the aggregator posts prices for each time interval in the next day.
In choosing these prices, the aggregator aims to ensure that the utility provided by the electricity consumption is high for the users (i.e., they are satisfied), while maintaining low costs for providing that electricity.
It is also paramount that the prices be chosen such that the customers' consumption does not violate the physical limits of the grid to avoid service outages and repair costs.
Our approach to day-ahead RTP, based on the algorithms developed in this paper, achieves high utility and ensures that cost constraints and grid constraints are satisfied without knowing the specific flexibility or responsiveness of the customer's load beforehand.
This is illustrated in Fig. \ref{fig:dr_model}.
For the remainder of this section, we formulate a day-ahead RTP problem that is utility maximizing and safe with regard to system constraints, and then show, through simulation, that the SPR and SUM algorithms are effective for this problem.

\begin{figure}[!t]
    \centering
    \includegraphics[width=\linewidth]{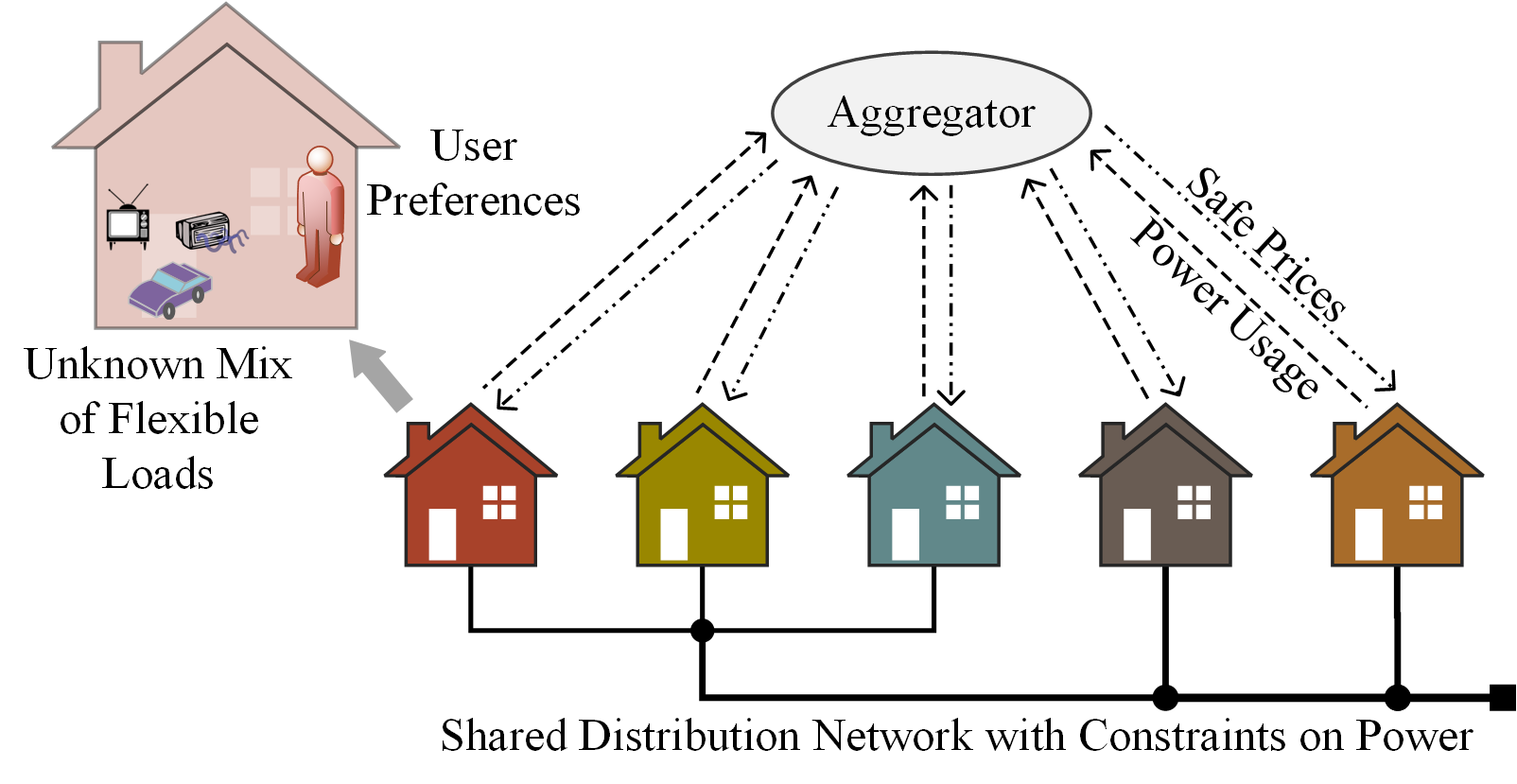}
    \vspace{-0.5cm}
    \caption{The aggregator does not specifically know how each electricity customer will respond to prices beforehand as each customer has an unknown mix of flexible loads and unique user preferences.
    Despite having limited knowledge as such, the aggregator needs to choose prices such that the utility is high and the distribution network constraints are satisfied.}
    \label{fig:dr_model}
\end{figure}

\subsection{Demand Response Formulation}

\label{sec:dr_form}
Each day (time step) $t$, the aggregator (central coordinator) posts prices $\gamma_{i,v}^t$ for each time period $v$ in $[V]$  and customer (user) $i$ in $[n]$.
Customer $i$ then responds with a noisy power consumption for each period in the day\footnote{Although the consumption $\bar{x}_i$ and price $\gamma_i^t$ are vectors in the demand response formulation, the results from the SPR and SUM formulation, where they are scalars, can be easily extended to this case. In stating our theoretical results, this vector case is not adopted for brevity of notation. }, denoted $\bar{x}_i(\gamma_i^t; \theta_i) = [ \bar{x}_{i,1}(\gamma_i^t; \theta_i)\ ...\ \bar{x}_{i,V}(\gamma_i^t; \theta_i)]^{\top}$ where $\bar{x}_{i,v}(\gamma_i^t; \theta_i) = x_{i,v}(\gamma_i^t; \theta_i) + \mu_{i,v}^t$ and $\gamma_i^t = [\gamma_{i,1}^t\ ...\ \gamma_{i,V}^t]$.
As before, we use a parametrically linear model for the consumption $x_{i,v}(\gamma_i^t; \theta_i) = h_{i,v} (\gamma_i^t)^{\top} \theta_i$ and take the noise $\mu_{i,v}^t$ to be conditionally subgaussian.
Note that the consumption at each period $v$ is allowed to depend on the price at all the periods in the day to account for inter-temporal flexibility.
The central coordinator observes the noisy consumption $\bar{x}_i(\gamma_i^t, \theta_i)$ for all customers on each day, and uses this to inform the choice of prices on future days.

When choosing prices, the aggregator aims to maximize utility while satisfying grid constraints and cost constraints.
We denote the utility function for customer $i$ as the increasing function $U_i:\mathbb{R}^V \rightarrow \mathbb{R}$.
The grid constraints are on the nodal voltages $u_v^t$ and the distribution line power flows $p_v^t$: $u_{min} \leq u_v^t \leq u_{max}$ and $p_v^t \leq S_{max}$ for all $v \in [V]$ and $t \in [T]$.
The nodal voltages and power flows are related to the consumption via the power flow model.
In particular, we use the \emph{LinDistFlow} model \cite{baran1989optimal} for a feeder network to express the reliability constraints linearly with respect to the consumption, of the form $\sum_{i=1}^n a_{ji} x_i (\gamma_i^t) \preceq \mathbf{1} c_j$ for all $j$ in $[p]$, $t$ in $[T]$.
Cost constraints can be implemented by specifying a limit on the total power supplied to the users at each period in the day according to the supply price and cost limit at that period.
With the objective and constraints defined, we have that the optimal prices satisfy
\begin{align}
\begin{split}
    \gamma^* \in & \argmax_{\gamma \in \mathbb{R}^{V \times n}} \ \sum_{i=1}^n U_i(x_i(\gamma_i))\\
    & \quad \text{s.t.}\ \sum_{i=1}^n a_{ji} x_i (\gamma_i) \preceq \mathbf{1} c_j, \ \forall j  \in [p]
\end{split}
\end{align}
In the next section, we discuss the price response model that is used to define the price response function.

\subsection{Price Response Model}

In order to capture the consumption behavior of a customer in response to electricity prices (i.e. specify $h_i$ in \eqref{eqn:price_resp}), we use the price response model developed in \cite{tucker2020constrained}, which itself uses the appliance model in \cite{alizadeh2014reduced}.
This appliance model considers clusters of appliance which each have a set of feasible consumption profiles (a consumption profile specifies the consumption from those appliances for each period in the day).
For example, one cluster might represent electric vehicles (EV) which need to be fully charged within a specific time frame subject to power limits.
Depending on how tight the time frame is, there might be several different consumption profiles for the electric vehicles that would satisfy these charging requirements.
Refer to \cite{alizadeh2014reduced} for further discussion on modeling other appliance types.

Given that there are several possible consumption profiles for each appliance cluster, the price response model in \cite{tucker2020constrained} considers two mechanisms by which price impacts a customer's power consumption: (1) the cost-minimizing appliance scheduling by the home energy management system (HEMS) and (2) the adjustment of the customer's preferences in response to electricity prices.
Mechanism (1) specifies that the HEMS will choose the consumption profile for each appliance cluster that minimizes the cost of operating that appliance while mechanism (2) specifies that the customer's usage of each appliance cluster varies according to price equally for all periods in the day.
We assume that the way in which the HEMS schedules appliances and the customer's preference adjustment function are known, while the number of appliances that each customer has in each appliance cluster (specified by each element of $\theta_i$) is unknown.
Note that our approach and algorithms could accommodate a more general model, but we use this one to provide an example of how the approach and algorithm can be used.

\subsection{Test Setup}

To evaluate the performance of our algorithms in the demand response problem through simulation, we use a real radial distribution network with $n=37$ customers as specified in \cite{tucker2020constrained} (originally from \cite{andrianesis2019locational}).
For this distribution system, we use the  power limits specified for each line given in \cite{tucker2020constrained} and the nodal voltage limits of 0.95 and 1.05 p.u. (with 12.5kV base) as given in \cite{andrianesis2019locational}.
We use $T = 365$ days and $V=3$ periods, with $m=2$ different appliance clusters.
One appliance cluster is for appliances that operate at the same time regardless of price and includes lighting (200 W, on for intervals $\{2,3\}$) and cooking (500 W, on for interval $3$).
The other appliance cluster is for flexible appliances that can be scheduled at several different times and includes EV charging (500 W, on for 1 interval in $\{ 1,3\}$), washer/drier (300 W, on for 1 interval in $\{ 2,3\}$), HVAC (600 W, on for 1 interval in $\{1,2,3\}$) and entertainment (200 W, on for 1 interval in $\{ 2,3\}$).
We use shifted sigmoids ($1/(1 + e^{\gamma_i^{\top} \mathbf{1} - 5})$) for the preference adjustment functions of the clusters.
Also, we use $U_i(x_i) = b_i \log (x_i + 1)$ where $b_i \sim \mathcal{U}[0,1]$ (for the SPR Algorithm), choose the unknown parameter $[\theta_i]_k \sim \mathcal{U}[0.5,1]$ for each $k$ in $[m]$, take the variance proxy as $\sigma = 1.5$ for the SPR experiments and $\sigma=3$ for the SUM experiments, where $\mu^t_{i,v} \sim \mathcal{N}(\sigma)$.
For the algorithm parameters, we use $\nu = 10$ for the regularization parameter, $T' = 52 \approx T^{2/3}$ for the duration of the safe exploration phase and $\delta = 0.01$ for the confidence parameter.
To make the algorithms tractable, we consider a finite set of prices in that the set of possible prices for each user at each time step is $\{2,5\}$.
We also divide the users in to three groups and use the same price for each of these groups to reduce the number of prices that need to be considered.
These groups are (1) users 11-25, (2) users 29-35, and (3) users 1-10, 26-28.
To reduce the computational load of the SUM experiments, the users in each group in the SUM experiments are assumed to have the same price response function.
We also calculate the SUM utility with a smoothed version of the price response as detailed in Appendix \ref{sec:apx_sumdr}.
Also, note that we do not incorporate the assumed bounds on $\theta_i$ (i.e. $\| \theta_i \| \leq S$, $\mathbf{1}^{\top} \theta_i \geq \rho$) when calculating the confidence sets for $\theta_i$.

\subsection{Simulation Results}

\begin{figure}[!t]
    \centering
    \includegraphics[width=\linewidth]{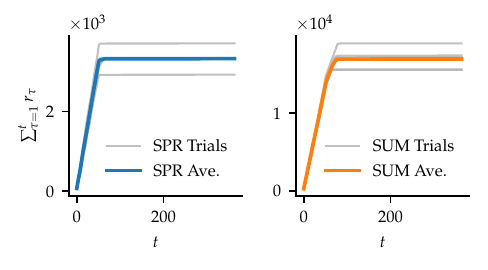}
    \vspace{-0.8cm}
    \caption{The rolling sum of the instantaneous regret of the Safe Price Response (SPR) and Safe Utility Maximization (SUM) algorithms for five trials of the demand response problem.}
    \label{fig:spr_reg}
\end{figure}

The SPR Algorithm (Algorithm \ref{alg:safe_price}) was implemented for the specified demand response problem and simulated for five trials, each with different realizations of the noise.
For all five trials, there were zero constraint violations.
The rolling sum of the instantaneous regret for each trial is shown in the left-hand side of Fig. \ref{fig:spr_reg}.

Similarly, the SUM Algorithm (Algorithm \ref{alg:net_opt}) was implemented for the demand response problem.
In five trials, there were zero constraint violations.
The rolling sum of the instantaneous regret for each trial is shown in the right-hand side of Fig. \ref{fig:spr_reg}.

\bltxt{Fig. \ref{fig:spr_reg} provides experimental evidence that the regret of both algorithms grows sub-linearly with respect to time.
This provides validation of the stated theoretical results, and demonstrates the algorithms' performance in a realistic demand response setting.}

\section{Conclusion}

In this paper, we present two novel safe optimization problems with applications to pricing design for safety-critical infrastructure systems.
We propose algorithms for each of these problems and prove in our analysis that they both enjoy sublinear regret.
We then demonstrate the real-world applicability of these algorithms by simulating them being used for DR pricing in a distribution network.
These simulations also provide numerical validation for the theoretical results in the paper.

\begin{bl}
    Despite the efficacy of our approach, there are some limitations.
    Firstly, our problem formulation requires that a set of safe prices is initially known by the algorithm.
    This is a fundamental limitation of any problem formulation with uncertain constraints that need to be satisfied at every time step because, in such a setting, the algorithm cannot ensure safety in the initial rounds without prior information.
    That said, it may be challenging in some real-world settings to accurately determine a set of safe prices without some (potentially unsafe) experimentation.

    Another limitation of our approach is that it only considers linear constraints.
    This limits the applicability of the work to safe pricing settings with nonlinear constraints.
    An important such example is power flow constraints in the smart grid, which can be more accurately specified by nonlinear constraints.
    We leave the problem of safe pricing under nonlinear constraints as future work.
\end{bl}

\bibliographystyle{IEEEtran}
\bibliography{references}

\begin{thebibliography}{10}
\providecommand{\url}[1]{#1}
\csname url@samestyle\endcsname
\providecommand{\newblock}{\relax}
\providecommand{\bibinfo}[2]{#2}
\providecommand{\BIBentrySTDinterwordspacing}{\spaceskip=0pt\relax}
\providecommand{\BIBentryALTinterwordstretchfactor}{4}
\providecommand{\BIBentryALTinterwordspacing}{\spaceskip=\fontdimen2\font plus
\BIBentryALTinterwordstretchfactor\fontdimen3\font minus
  \fontdimen4\font\relax}
\providecommand{\BIBforeignlanguage}[2]{{%
\expandafter\ifx\csname l@#1\endcsname\relax
\typeout{** WARNING: IEEEtran.bst: No hyphenation pattern has been}%
\typeout{** loaded for the language `#1'. Using the pattern for}%
\typeout{** the default language instead.}%
\else
\language=\csname l@#1\endcsname
\fi
#2}}
\providecommand{\BIBdecl}{\relax}
\BIBdecl

\bibitem{samadi2010optimal}
P.~Samadi, A.-H. Mohsenian-Rad, R.~Schober, V.~W. Wong, and J.~Jatskevich,
  ``Optimal real-time pricing algorithm based on utility maximization for smart
  grid,'' in \emph{2010 First IEEE International Conference on Smart Grid
  Communications}.\hskip 1em plus 0.5em minus 0.4em\relax IEEE, 2010, pp.
  415--420.

\bibitem{mehr2017joint}
N.~Mehr, J.~Lioris, R.~Horowitz, and R.~Pedarsani, ``Joint perimeter and signal
  control of urban traffic via network utility maximization,'' in \emph{2017
  IEEE 20th International Conference on Intelligent Transportation Systems
  (ITSC)}.\hskip 1em plus 0.5em minus 0.4em\relax IEEE, 2017, pp. 1--6.

\bibitem{nedic2009distributed}
A.~Nedic and A.~Ozdaglar, ``Distributed subgradient methods for multi-agent
  optimization,'' \emph{IEEE Transactions on Automatic Control}, vol.~54,
  no.~1, pp. 48--61, 2009.

\bibitem{low1999optimization}
S.~H. Low and D.~E. Lapsley, ``Optimization flow control. i. basic algorithm
  and convergence,'' \emph{IEEE/ACM Transactions on networking}, vol.~7, no.~6,
  pp. 861--874, 1999.

\bibitem{palomar2007alternative}
D.~P. Palomar and M.~Chiang, ``Alternative distributed algorithms for network
  utility maximization: Framework and applications,'' \emph{IEEE Transactions
  on Automatic Control}, vol.~52, no.~12, pp. 2254--2269, 2007.

\bibitem{bitar2012deadline}
E.~Bitar and S.~Low, ``Deadline differentiated pricing of deferrable electric
  power service,'' in \emph{2012 IEEE 51st IEEE conference on decision and
  control (CDC)}.\hskip 1em plus 0.5em minus 0.4em\relax IEEE, 2012, pp.
  4991--4997.

\bibitem{han2015approximately}
S.~Han, U.~Topcu, and G.~J. Pappas, ``An approximately truthful mechanism for
  electric vehicle charging via joint differential privacy,'' in \emph{2015
  American Control Conference (ACC)}.\hskip 1em plus 0.5em minus 0.4em\relax
  IEEE, 2015, pp. 2469--2475.

\bibitem{sun2018eliciting}
B.~Sun, X.~Tan, and D.~H. Tsang, ``Eliciting multi-dimensional flexibilities
  from electric vehicles: A mechanism design approach,'' \emph{IEEE
  Transactions on Power Systems}, vol.~34, no.~5, pp. 4038--4047, 2018.

\bibitem{palomar2006tutorial}
D.~P. Palomar and M.~Chiang, ``A tutorial on decomposition methods for network
  utility maximization,'' \emph{IEEE Journal on Selected Areas in
  Communications}, vol.~24, no.~8, pp. 1439--1451, 2006.

\bibitem{chiang2007layering}
M.~Chiang, S.~H. Low, A.~R. Calderbank, and J.~C. Doyle, ``Layering as
  optimization decomposition: A mathematical theory of network architectures,''
  \emph{Proceedings of the IEEE}, vol.~95, no.~1, pp. 255--312, 2007.

\bibitem{li2011optimal}
N.~Li, L.~Chen, and S.~H. Low, ``Optimal demand response based on utility
  maximization in power networks,'' in \emph{2011 IEEE power and energy society
  general meeting}.\hskip 1em plus 0.5em minus 0.4em\relax IEEE, 2011, pp.
  1--8.

\bibitem{turan2022safe}
B.~Turan and M.~Alizadeh, ``Safe dual gradient method for network utility
  maximization problems,'' \emph{arXiv preprint arXiv:2208.04446}, 2022.

\bibitem{nedic2014distributed}
A.~Nedi{\'c} and A.~Olshevsky, ``Distributed optimization over time-varying
  directed graphs,'' \emph{IEEE Transactions on Automatic Control}, vol.~60,
  no.~3, pp. 601--615, 2014.

\bibitem{brunke2022safe}
L.~Brunke, M.~Greeff, A.~W. Hall, Z.~Yuan, S.~Zhou, J.~Panerati, and A.~P.
  Schoellig, ``Safe learning in robotics: From learning-based control to safe
  reinforcement learning,'' \emph{Annual Review of Control, Robotics, and
  Autonomous Systems}, vol.~5, pp. 411--444, 2022.

\bibitem{gahlawat2020l1}
A.~Gahlawat, P.~Zhao, A.~Patterson, N.~Hovakimyan, and E.~Theodorou, ``L1-gp:
  L1 adaptive control with bayesian learning,'' in \emph{Learning for Dynamics
  and Control}.\hskip 1em plus 0.5em minus 0.4em\relax PMLR, 2020, pp.
  826--837.

\bibitem{chowdhary2014bayesian}
G.~Chowdhary, H.~A. Kingravi, J.~P. How, and P.~A. Vela, ``Bayesian
  nonparametric adaptive control using gaussian processes,'' \emph{IEEE
  transactions on neural networks and learning systems}, vol.~26, no.~3, pp.
  537--550, 2014.

\bibitem{berkenkamp2015safe}
F.~Berkenkamp and A.~P. Schoellig, ``Safe and robust learning control with
  gaussian processes,'' in \emph{2015 European Control Conference (ECC)}.\hskip
  1em plus 0.5em minus 0.4em\relax IEEE, 2015, pp. 2496--2501.

\bibitem{koller2018learning}
T.~Koller, F.~Berkenkamp, M.~Turchetta, and A.~Krause, ``Learning-based model
  predictive control for safe exploration,'' in \emph{2018 IEEE conference on
  decision and control (CDC)}.\hskip 1em plus 0.5em minus 0.4em\relax IEEE,
  2018, pp. 6059--6066.

\bibitem{wabersich2018linear}
K.~P. Wabersich and M.~N. Zeilinger, ``Linear model predictive safety
  certification for learning-based control,'' in \emph{2018 IEEE Conference on
  Decision and Control (CDC)}.\hskip 1em plus 0.5em minus 0.4em\relax IEEE,
  2018, pp. 7130--7135.

\bibitem{usmanova2019safe}
I.~Usmanova, A.~Krause, and M.~Kamgarpour, ``Safe convex learning under
  uncertain constraints,'' in \emph{The 22nd International Conference on
  Artificial Intelligence and Statistics}.\hskip 1em plus 0.5em minus
  0.4em\relax PMLR, 2019, pp. 2106--2114.

\bibitem{usmanova2020safe}
------, ``Safe non-smooth black-box optimization with application to policy
  search,'' in \emph{Learning for Dynamics and Control}.\hskip 1em plus 0.5em
  minus 0.4em\relax PMLR, 2020, pp. 980--989.

\bibitem{chen2018bandit}
T.~Chen and G.~B. Giannakis, ``Bandit convex optimization for scalable and
  dynamic iot management,'' \emph{IEEE Internet of Things Journal}, vol.~6,
  no.~1, pp. 1276--1286, 2018.

\bibitem{chaudhary2021safe}
S.~Chaudhary and D.~Kalathil, ``Safe online convex optimization with unknown
  linear safety constraints,'' \emph{arXiv preprint arXiv:2111.07430}, 2021.

\bibitem{sui2015safe}
Y.~Sui, A.~Gotovos, J.~Burdick, and A.~Krause, ``Safe exploration for
  optimization with gaussian processes,'' in \emph{International conference on
  machine learning}.\hskip 1em plus 0.5em minus 0.4em\relax PMLR, 2015, pp.
  997--1005.

\bibitem{berkenkamp2021bayesian}
F.~Berkenkamp, A.~Krause, and A.~P. Schoellig, ``Bayesian optimization with
  safety constraints: safe and automatic parameter tuning in robotics,''
  \emph{Machine Learning}, pp. 1--35, 2021.

\bibitem{moradipari2020stage}
A.~Moradipari, C.~Thrampoulidis, and M.~Alizadeh, ``Stage-wise conservative
  linear bandits,'' \emph{Advances in neural information processing systems},
  vol.~33, pp. 11\,191--11\,201, 2020.

\bibitem{pacchiano2021stochastic}
A.~Pacchiano, M.~Ghavamzadeh, P.~Bartlett, and H.~Jiang, ``Stochastic bandits
  with linear constraints,'' in \emph{International Conference on Artificial
  Intelligence and Statistics}.\hskip 1em plus 0.5em minus 0.4em\relax PMLR,
  2021, pp. 2827--2835.

\bibitem{amani2019linear}
S.~Amani, M.~Alizadeh, and C.~Thrampoulidis, ``Linear stochastic bandits under
  safety constraints,'' \emph{Advances in Neural Information Processing
  Systems}, vol.~32, 2019.

\bibitem{cdc2022paper}
S.~Hutchinson, B.~Turan, and M.~Alizadeh, ``A safe pricing mechanism for
  distributed resource allocation with bandit feedback,'' in \emph{2022 61st
  IEEE Conference on Decision and Control (CDC)}.\hskip 1em plus 0.5em minus
  0.4em\relax IEEE, 2022.

\bibitem{abbasi2011improved}
Y.~Abbasi-Yadkori, D.~P{\'a}l, and C.~Szepesv{\'a}ri, ``Improved algorithms for
  linear stochastic bandits,'' \emph{Advances in neural information processing
  systems}, vol.~24, 2011.

\bibitem{kelly1998rate}
F.~P. Kelly, A.~K. Maulloo, and D.~K.~H. Tan, ``Rate control for communication
  networks: shadow prices, proportional fairness and stability,'' \emph{Journal
  of the Operational Research society}, vol.~49, no.~3, pp. 237--252, 1998.

\bibitem{mo2000fair}
J.~Mo and J.~Walrand, ``Fair end-to-end window-based congestion control,''
  \emph{IEEE/ACM Transactions on networking}, vol.~8, no.~5, pp. 556--567,
  2000.

\bibitem{lan2010axiomatic}
T.~Lan, D.~Kao, M.~Chiang, and A.~Sabharwal, \emph{An axiomatic theory of
  fairness in network resource allocation}.\hskip 1em plus 0.5em minus
  0.4em\relax IEEE, 2010.

\bibitem{dani2008stochastic}
V.~Dani, T.~P. Hayes, and S.~M. Kakade, ``Stochastic linear optimization under
  bandit feedback,'' 2008.

\bibitem{qdr2006benefits}
Q.~Qdr, ``Benefits of demand response in electricity markets and
  recommendations for achieving them,'' \emph{US Dept. Energy, Washington, DC,
  USA, Tech. Rep}, vol. 2006, 2006.

\bibitem{baran1989optimal}
M.~E. Baran and F.~F. Wu, ``Optimal capacitor placement on radial distribution
  systems,'' \emph{IEEE Transactions on power Delivery}, vol.~4, no.~1, pp.
  725--734, 1989.

\bibitem{tucker2020constrained}
N.~Tucker, A.~Moradipari, and M.~Alizadeh, ``Constrained thompson sampling for
  real-time electricity pricing with grid reliability constraints,'' \emph{IEEE
  Transactions on Smart Grid}, vol.~11, no.~6, pp. 4971--4983, 2020.

\bibitem{alizadeh2014reduced}
M.~Alizadeh, A.~Scaglione, A.~Applebaum, G.~Kesidis, and K.~Levitt,
  ``Reduced-order load models for large populations of flexible appliances,''
  \emph{IEEE Transactions on Power Systems}, vol.~30, no.~4, pp. 1758--1774,
  2014.

\bibitem{andrianesis2019locational}
P.~Andrianesis, M.~Caramanis, R.~D. Masiello, R.~D. Tabors, and S.~Bahramirad,
  ``Locational marginal value of distributed energy resources as non-wires
  alternatives,'' \emph{IEEE Transactions on Smart Grid}, vol.~11, no.~1, pp.
  270--280, 2019.

\bibitem{lattimore2020bandit}
T.~Lattimore and C.~Szepesv{\'a}ri, \emph{Bandit algorithms}.\hskip 1em plus
  0.5em minus 0.4em\relax Cambridge University Press, 2020.

\bibitem{boyd2004convex}
S.~Boyd, S.~P. Boyd, and L.~Vandenberghe, \emph{Convex optimization}.\hskip 1em
  plus 0.5em minus 0.4em\relax Cambridge university press, 2004.

\bibitem{bertsekas2009convex}
D.~Bertsekas, \emph{Convex optimization theory}.\hskip 1em plus 0.5em minus
  0.4em\relax Athena Scientific, 2009, vol.~1.

\bibitem{keisler1996mathematical}
\BIBentryALTinterwordspacing
H.~Keisler and J.~Robbin, \emph{Mathematical Logic and Computability}, ser.
  International series in pure and applied mathematics.\hskip 1em plus 0.5em
  minus 0.4em\relax McGraw-Hill, 1996. [Online]. Available:
  \url{https://books.google.com/books?id=4r2eQgAACAAJ}
\BIBentrySTDinterwordspacing

\end{thebibliography}

\appendix

\subsection{Proof of Theorem \ref{thm:reg_bound}}
\label{sec:apx_thm2}
We decompose the instantaneous regret in to $r_t^I$ and $r_t^{II}$ as
\begin{alignat*}{2}
    r_t = & \sum_{i=1}^n [ f_i (x_i (\gamma_i^*; \theta_i^*)) - f_i (x_i (\gamma_i^t;\theta_i^*))]\\
    = & \sum_{i=1}^n [ f_i (h_i(\gamma_i^*)^{\top} \theta_i^*) - f_i (h_i(\gamma_i^t)^{\top} \tilde{\theta_i}) ]&\quad & \Big \}\ r_t^I\\
     & + \sum_{i=1}^n [ f_i (h_i(\gamma_i^t)^{\top} \tilde{\theta_i}) - f_i (h_i(\gamma_i^t)^{\top} \theta_i^* ) ].&& \Big \}\ r_t^{II}
\end{alignat*}
In the following sections we handle $r_t^I$ and $r_t^{II}$ separately.

\subsubsection{Bounding $r_t^I$ (proof of Lemma \ref{lem:term1})}
\label{sec:apx_rt1}
Note that, in this section, many of the statements only hold with high probability (i.e. they rely on Theorem \ref{thm:conf_set} and Lemma \ref{lem:min_eig}).
For brevity, this is not referenced at each step but the probability of the complete bound is discussed at the end.

First, we define some sets that allow us to more easily bound the growth of the safe decision set.
We first have an expanded confidence set for $\theta_i^*$ that is centered at $\theta_i^*$ instead of $\hat{\theta}_i$ :
\begin{equation}
\label{eqn:exp_conf}
    \tilde{C}_i^t = \{ \theta_i \in \mathbb{R}^m_+ : \| \theta_i - \theta_i^* \|_{V_i^t} \leq 2 \sqrt{\beta^t}, \| \theta_i \| \leq S \}.
\end{equation}
Note that $C_i^t \subseteq \tilde{C}_i^t$ as we can use the triangle inequality with any $\theta_i \in C_i^t$ to get $\| \theta_i - \theta_i^* \|_{V_i^t} \allowbreak = \| \theta_i - \hat{\theta}_i + \hat{\theta}_i - \theta_i^* \|_{V_i^t} \allowbreak \leq \| \theta_i - \hat{\theta}_i\|_{V_i^t} + \| \hat{\theta}_i - \theta_i^* \|_{V_i^t} \allowbreak \leq 2 \sqrt{\beta^t}$.
We then use this expanded confidence set to define a shrunk safe price set,
\begin{equation}
\label{eqn:shr_pset}
    \tilde{D}^t = \{ \gamma \in \mathbb{R}^n : \sum_{i=1}^n a_{ji} x_i(\gamma_i; \theta_i) \leq c_j, \forall \theta \in \tilde{C}^t, \forall j \in [p]\},
\end{equation}
where 
\begin{equation}
\label{eqn:lconf_set}
    \tilde{C}^t = \tilde{C}_1^t \times \tilde{C}_2^t \times ... \times \tilde{C}_n^t.
\end{equation}
The remaining analysis deals with the $h$ vector, which is defined as $h = [h_1(\gamma_1)^{\top} \ h_2(\gamma_2)^{\top} \ ...\ h_n(\gamma_n)^{\top} ]^{\top}$.
Accordingly, we define a safe set for the $h$ vector:
\begin{equation}
\label{eqn:g_tilde}
    \tilde{G}^t = \{ [h_1(\gamma_1)^{\top}\ ...\ h_n(\gamma_n)^{\top}]^{\top} : [\gamma_1\ ...\ \gamma_n]^{\top} \in \tilde{D}^t \}.
\end{equation}
We also have the initial safe set for the $h$ vector, which is defined similarily:
\begin{equation}
\label{eqn:g_init}
    G^0 = \{ [h_1(\gamma_1)^{\top}\ ...\ h_n(\gamma_n)^{\top}]^{\top} : [\gamma_1\ ...\ \gamma_n]^{\top} \in D^0 \}
\end{equation}
Note that, by definition, $\tilde{C}^t \subseteq C^0$ for all $t \geq 1$ which implies that $D^0 \subseteq \tilde{D}^t$ for all $t \geq 1$ and equivalently  that $G^0 \subseteq \tilde{G}^t$ for all $t \geq 1$.

Next, we consider a line from a point in $G^0$ to the optimal $h$ vector.
Let $h^* = [h_1(\gamma_1^*)^{\top}\ ...\ h_n(\gamma_n^*)^{\top}]^{\top}$ and $h^0$ be any element in $G^0$.
Then, we can use $\alpha^t$ to track the safe set along a line as
\begin{equation}
\label{eqn:alpha}
    \alpha^t = \max \{ \alpha \in [0,1] : \alpha h^* + (1 - \alpha) h^0 \in \tilde{G}^t \}.
\end{equation}
Let $z^t = \alpha^t h^* + (1 - \alpha^t) h^0$ and $z_i^t = \alpha^t h_i^* + (1 - \alpha^t) h_i^0$.
We can then bound $r_{t+1}^I$ with $\alpha^t$ by using the fact that $\sum_{i=1}^n f_i (h_i(\gamma_i^{t+1})^{\top} \tilde{\theta}_i^{t+1}) \geq \sum_{i=1}^n f_i ([z_i^t]^{\top}\theta_i^*)$ (which follows from \eqref{eqn:ofu_price}). 
\begin{align*}
    r_{t+1}^I & = \sum_{i=1}^n [ f_i (h_i(\gamma_i^*)^{\top}\theta_i^*) - f_i (h_i(\gamma_i^{t+1})^{\top} \tilde{\theta}_i^{t+1}) ]\\
    & \leq \sum_{i=1}^n | f_i (h_i(\gamma_i^*)^{\top}\theta_i^*) - f_i ([z_i^t]^{\top}\theta_i^*) | \\
    & \leq M \sum_{i=1}^n | h_i(\gamma_i^*)^{\top}\theta_i^* - [z_i^t]^{\top}\theta_i^* | \\
    & = M \sum_{i=1}^n | (h_i^* - h_i^0) ^{\top} \theta_i^* | (1 - \alpha^t) \\
    & \leq M \sum_{i=1}^n  \| h_i^* - h_i^0 \| \| \theta_i^* \| (1 - \alpha^t) \\
    & \leq 2 M n L S (1 - \alpha^t) 
\end{align*}
At this point, we can see that establishing a lower bound on $\alpha^t$ will provide an upper bound on $r_t^I$.

To do so, we establish that at least one constraint is tight on the optimal solution and then use this property to lower bound $\alpha^t$.
We first show that the problem can be expressed as optimization over the consumption in Lemma \ref{lem:equiv_sol}.
\begin{lemma}
\label{lem:equiv_sol}
    Let Assumption \labelcref{ass:price_resp} hold. Note the definition of $\gamma^*$ in \eqref{eqn:opt_price}. Then the optimal consumption $x^* = [h_1(\gamma_1^*)^{\top}\theta_1^*\ ...\ h_n(\gamma_n^*)^{\top} \theta_n^*]^{\top}$ satisfies
    \begin{equation}
    \label{eqn:equiv_sol}
    x^* \in \argmax_{x \in E} \sum_{i=1}^n f_i (x_i),
    \end{equation}
    where
    \begin{equation*}
    E = \left\{ x\in \mathbb{R}^n_{++} :  \sum_{i=1}^n a_{ji} x_i \leq c_j,\ \forall j \in [p] \right\}.
    \end{equation*}
\end{lemma}
\begin{proof}
Consider the set
\begin{equation*}
    \tilde{E} = \left\{ \left[h_1(\gamma_1)^{\top}\theta_1^*\ ...\ h_n(\gamma_n)^{\top} \theta_n^* \right]^{\top} : \gamma \in \bar{D} \right\}.
\end{equation*}
We can see that $\argmax_{x \in \tilde{E}} \sum_{i=1}^n f_i (x_i)$ contains $[h_1(\gamma_1^*)^{\top}\theta_1^*\ ...\ h_n(\gamma_n^*)^{\top} \theta_n^*]^{\top}$, where $\gamma^*$ is defined in \eqref{eqn:opt_price}.
Therefore, it only remains to be shown that $\tilde{E}$ is equal to $E$.
To do so, we consider the range of values that \bltxt{$h_i(\gamma_i)^{\top}\theta_i^*$} can take for $\gamma \in \bar{D}$ and hence the values that are in $\tilde{E}$.
% From Assumption \ref{ass:price_resp}, we have that for all $i \in n$ there exists $y_1$ and $y_2$ in $[-\infty, \infty]$ such that $h_i(y_1) = \infty$ and $h_i(y_2) = 0$.
% Since $\theta_i^*$ is nonzero,  there exists $y_1$ and $y_2$ in $[-\infty, \infty]$ such that $h_i(y_1)^T \theta_i^* = \infty$ and $h_i(y_2)^T \theta_i^* = 0$.
\bltxt{Since $h_i(y)$ is continuous, $h_i(y)^T \theta_i^*$ is continuous as a function of $y$.
Also, Assumption \ref{ass:price_resp} and the fact that $\theta_i^*$ is nonzero imply that $\lim_{y \rightarrow -\infty} h_i(y)^T \theta_i^* = \infty$ and that $\lim_{y \rightarrow \infty} h_i(y)^T \theta_i^* = 0$.
Therefore, by the Intermediate Value Theorem, for every $x_i \in \mathbb{R}_+$, there exists a $y \in \mathbb{R}$ such that $x_i = h_i(y)^T \theta_i^*$.}
Thus, $\tilde{E}$ contains all $x \succ 0$ such that $\sum_{i=1}^n a_{ji} x_i \leq c_j$ for all $j$ in $[p]$ and is therefore equal to $E$.
This completes the proof.
\end{proof}

Using this result, we then show that at least one constraint is tight on the optimal solution in Lemma \ref{lem:tight_const}.
\begin{lemma}
\label{lem:tight_const}
Let Assumption \ref{ass:price_resp} hold. Then, there exists a constraint $j$ in $[p]$ such that $\sum_{i=1}^n a_{ji} x_i^* = c_j$, where $x^*$ is an optimal consumption as defined in Lemma \ref{lem:equiv_sol}.
\end{lemma}
\begin{proof}
The statement of the Lemma is equivalent to claiming that $x^*$ is in the boundary of $F$ where 
\begin{equation*}
    F = \left\{x \in \mathbb{R}^n : \sum_{i=1}^n a_{ji} x_i \leq c_j\ \forall j \in [p] \right\}.
\end{equation*}
First, note that $x^*$ is in $F$ by definition, as $E \subseteq F$.
Suppose that $x^*$ is not in the boundary of $F$, which implies that $x^*$ is in the interior of $F$.
This would mean that there exists an open ball of radius $\epsilon > 0$ centered at $x^*$ that is in $F$.
It would follow that $x = x^* + \frac{\epsilon}{2} \frac{\mathbf{1}}{\| \mathbf{1} \|} \in F$.
Therefore, $x \succ x^*$.
It follows from the fact that $f_i$ is increasing that $f_i (x_i) > f_i (x_i^*)$ for all $i \in [n]$, which implies that $\sum_{i=1}^n f_i (x_i) > \sum_{i=1}^n f_i (x_i^*)$.
This means that $x^*$ does not satisfy \eqref{eqn:equiv_sol}, which is a contradiction.
Therefore, $x^*$ is in the boundary of $F$, which completes the proof.
\end{proof}

We can then use the fact that the constraint is tight on the optimal solution to establish a property of $\alpha^t$ that ultimately allows us to bound $\alpha^t$.
\begin{lemma}
\label{lem:z_tight}
    Let Assumption \ref{ass:price_resp} hold. Recall the definition of $\alpha^t$ in \eqref{eqn:alpha} and the definition of $\tilde{C}^t$ in \eqref{eqn:lconf_set}. Also, let $z_i^t = \alpha^t h_i^* + (1 - \alpha^t) h_i^0$ for all $i$ in $[n]$.
    Then with probability at least $1 - \delta$, there exists $j$ in $[p]$ such that
    \begin{equation}
    \label{eqn:z_tight}
        \max_{\theta \in \tilde{C}^t} \sum_{i=1}^n a_{ji} \theta_i^{\top} z_i^t = c_j.
    \end{equation}
\end{lemma}
\begin{proof}
Recall the definition of $\tilde{G}^t$ in \eqref{eqn:g_tilde}.
The following statements hold with probability at least $1 - \delta$ without being further referenced.
Also, let $H^t := \{ h \in \mathbb{R}^{nm} | \sum_{i=1}^n a_{ji} h_i^T \theta_i \leq c_j\ \forall j \in [p], \forall \theta \in \tilde{C}^t \}$.
Note that $\tilde{G}^t = H^t \cap \{ h \in \mathbb{R}^{nm} : h \succeq 0 \}$.

First we show that $h^*$ is not in the $\mathbf{int} H^t$.
It is known that $\theta^*  \in \tilde{C}^t$.
If $\tilde{C}^t = \{ \theta^* \}$, then $\mathbf{int} H^t = \{ [h_1^T\ h_2^T\ ...\ h_n^T]^T | \sum_{i=1}^n a_{ji} h_i^T \theta_i^* < c_j\ \forall j \in [p] \}$.
From Lemma \labelcref{lem:equiv_sol,lem:tight_const}, we have that there is a $j \in [p]$ such that $\sum_{i=1}^n a_{ji} h_i^{*,T} \theta_i^* = c_j$, so $h^*$ is not in $\mathbf{int} H^t$.
Alternatively, if $\{\theta^*\} \subset \tilde{C}^t $, then $\mathbf{int}H^t \subset \{[h_1^T\ h_2^T\ ...\ h_n^T]^T | \sum_{i=1}^n a_{ji} h_i^T \theta_i^*  < c_j\ \forall j \in [p]\}$ which $h^*$ cannot be in either.
Therefore, $h^* \notin \mathbf{int} H$.

Next, we show that $z^t$ is in the boundary of $H^t$.
Note that $z^t \in H^t$ as $z^t \in \tilde{G}^t$ and $\tilde{G}^t \subseteq H^t$.
First, consider the case where $h^* \in \mathbf{bd}H^t$.
This will result in $\alpha^t = 1$ and $z^t = h^* \in \mathbf{bd}H^t$.
Next, consider the case where $h^* \notin H^t$.
For this case, suppose that $z^t \in \mathbf{int} H^t$ such that there exists an open ball of radius $\epsilon > 0$ centered at $z^t$ that is in $H^t$.
It then follows that $z^t + \frac{\epsilon}{2} \frac{h^* - h^0}{\| h^* - h^0 \|} \in H^t$.
This can be manipulated as follows
\begin{align*}
    z^t + \frac{\epsilon}{2} \frac{h^* - h^0}{\| h^* - h^0 \|} & = h^0 + \alpha^t(h^* - h^0) + \frac{\epsilon}{2} \frac{h^* - h^0}{\| h^* - h^0 \|} \\
    & = h^0 + \bigg(\alpha^t + \frac{\epsilon}{2\| h^* - h^0 \|} \bigg) (h^* - h^0)\\
    & = h^0 + \tilde{\alpha}^t (h^* - h^0),
\end{align*}
where $\tilde{\alpha}^t = (\alpha^t + \frac{\epsilon}{2\| h^* - h^0 \|} ) > \alpha^t$.
Since $h^* \notin H^t$, $\tilde{\alpha}^t < 1$.
Also, note that since $h^0 \succeq 0$ and $h^* \succeq 0$, we also have that $h^0 + \tilde{\alpha}^t (h^* - h^0) \succeq 0$.
This means that $h^0 + \tilde{\alpha}^t (h^* - h^0) \in \tilde{G}^t$.
Because there exists an $\tilde{\alpha}^t \in (0,1)$ such that $h^0 + \tilde{\alpha}^t (h^* - h^0) \in \tilde{G}^t$ and $\tilde{\alpha}^t > \alpha^t$, the definition of $\alpha^t$ does not hold.
Therefore, $z^t \in \mathbf{bd} H^t$.
The statement of the lemma immediately follows.
\end{proof}

To lower bound $\alpha^t$, we use Lemma \ref{lem:z_tight} which gives that there is a $j$ in $[p]$ such that $ \max_{\theta \in \tilde{C}^t} \sum_{i=1}^n a_{ji} \theta_i^T z_i^t = c_j$.
Therefore,
\begin{align*}
    c_j & = \max_{\theta \in \tilde{C}^t} \sum_{i=1}^n a_{ji} [z_i^t]^{\top}\theta_i\\
    & = \sum_{i=1}^n \max_{\theta_i \in \tilde{C}^t_i} a_{ji} [z_i^t]^{\top}\theta_i\\
    & = \sum_{i=1}^n \max_{\theta_i \in \tilde{C}^t_i} a_{ji} (\alpha^t h_i^* + (1 - \alpha^t) h_i^0)^{\top} \theta_i\\
    & \leq \sum_{i=1}^n \max_{\theta_i \in \tilde{C}^t_i} a_{ji} \alpha^t [h_i^*]^{\top} \theta_i + \sum_{i=1}^n \max_{\theta_i \in \tilde{C}^t_i} a_{ji} (1 - \alpha^t) [h_i^0]^{\top} \theta_i\\
    & = \alpha^t \underbrace{\sum_{i=1}^n \max_{\theta_i \in \tilde{C}^t_i} a_{ji} [h_i^*]^{\top} \theta_i}_{b^*} + (1 - \alpha^t) \underbrace{\sum_{i=1}^n \max_{\theta_i \in \tilde{C}^t_i} a_{ji} [h_i^0]^{\top} \theta_i}_{b^0} \numberthis \label{eqn:sep_cons}
\end{align*}

Next, we bound $b^*$.
In order to do so, we need the following lemmas, together which lower bound the minimum eigenvalue of the gram matrix after the pure exploration phase.

\begin{lemma}
\label{lem:pos_def}
    (Duplicate of Lemma \ref{lem:pos_def1}) Let Assumption \ref{ass:init_set} hold. Then, with Algorithm \ref{alg:safe_price} we have that 
    \begin{equation}
    % \label{eqn:min_eig}
        \lambda_- := \min_{i \in [n]} \bigg[ \lambda_{min} \Big(\mathbb{E} \left[ h_i(\gamma_i^t) h_i(\gamma_i^t)^{\top} \right] \Big) \bigg] > 0,
    \end{equation}
    for all $t$ in $[1,T']$.
\end{lemma}
\begin{proof} 
Let $\ell = \ell(\gamma_i^t) = h_i(\gamma_i^t)$.
First note that for any random vector $Y$, the form $\mathbb{E}[Y Y^\top]$ is positive semi-definite and therefore $\mathbb{E}[\ell \ell^\top]$ is positive semi-definite by definition.
To see this, we have for all vectors $b$ that $b^\top \mathbb{E}[Y Y^\top] b = \mathbb{E}[b^\top Y Y^\top b] = \mathbb{E}[(b^\top Y)^2] \geq 0$.
We can then proceed by supposing that $\mathbb{E}[\ell \ell^\top]$ is not positive definite.
This requires that there is some nonzero $a \in \mathbb{R}^m$ such that $a^\top \mathbb{E}[\ell \ell^\top] a = 0$.
It would follow that $\mathbb{E}[(a^\top \ell)^2] = 0$ and therefore that $a^\top \ell = 0$ almost surely. 
From Assumption \ref{ass:init_set}, we have that for all $i \in [n]$ there is no nonzero $\alpha_i \in \mathbb{R}^m$ such that $\alpha_i^\top h_i(y_i) = 0$ for all $[y_1\ y_2\ ...\ y_n]^{\top} \in D^0$.
Since $\gamma^t \sim \mathrm{Unif}(D^0)$, there is no nonzero $a \in \mathbb{R}^m$ such that $a^\top \ell (\gamma_i^t) = 0$ for any $i$ in $[n]$ almost surely.
This demonstrates that $\mathbb{E}[\ell \ell^\top]$ is positive definite for all $i \in [n]$ and hence the statement of the lemma holds.
\end{proof}

\begin{lemma}
\label{lem:min_eig}
    (Lemma 1 in \cite{amani2019linear} modified for multiple users) With $V_i^t$ as defined in \eqref{eqn:gram_mat} and $\lambda_-$ as defined in \eqref{eqn:min_eig}, we have with probability at least $1 - \delta$ that
    \begin{equation}
        \lambda_{\text{min}} (V_i^{T'}) \geq \nu + \frac{\lambda_- T'}{2},
    \end{equation}
    for $T' \geq t_\delta := \frac{8 L^2}{\lambda_-} \log(\frac{n m}{\delta})$ and all $i$ in $[n]$.
\end{lemma}

We can then establish a bound on $b^*$ for $T \geq t+1 > T' \geq t_{\delta}$ as follows.
\begin{align*}
    b^* & = \sum_{i=1}^n \max_{\theta_i \in \tilde{C}^t_i} [a_{ji} h_i^*]^{\top} \theta_i\\
    & \overset{\text{(a)}}{\leq} \sum_{i=1}^n \left( a_{ji} [h_i^*]^{\top} \theta_i^* + 2 | a_{ji} | \sqrt{\beta^t} \| h_i^* \|_{[V_i^t]^{-1}} \right)\\
    & \overset{\text{(b)}}{=}  c_j + \sum_{i=1}^n 2 | a_{ji} | \sqrt{\beta^t} \| h_i^* \|_{[V_i^t]^{-1}} \\
    & \leq  c_j + \sum_{i=1}^n \frac{2 | a_{ji} | \sqrt{\beta^t} \| h_i^* \|}{\sqrt{\lambda_{\text{min}}(V_i^t)}} \\
    & \overset{\text{(c)}}{\leq}  c_j + \sum_{i=1}^n \frac{2 | a_{ji} | \sqrt{\beta^t} \| h_i^* \|}{\sqrt{\lambda_{\text{min}}(V_i^{T'+1})}} \\
    & \overset{\text{(d)}}{\leq}  c_j + \sum_{i=1}^n \frac{2 | a_{ji} | \sqrt{2 \beta^t} \| h_i^* \|}{\sqrt{2 \nu + \lambda_- T'}} \\
    & \overset{\text{(e)}}{\leq}  c_j + \underbrace{\frac{2 n \kappa L \sqrt{2 \beta^T}}{\sqrt{2 \nu + \lambda_- T'}}}_{\ell^t} \numberthis \label{eqn:bstar}
\end{align*}
The step (a) is due to the closed form solution for the support function of an ellipsoid (e.g. \cite{lattimore2020bandit} Eq. 19.13), (b) is due to Lemma \ref{lem:z_tight}, (c) is due to the fact that $\lambda_{\text{min}}(V_i^t) \geq \lambda_{\text{min}}(V_i^{T'})$ for $t \geq T'$, (d) is due to Lemma \ref{lem:min_eig} and (e) is due to Assumptions \labelcref{ass:init_set,ass:price_resp}, and the fact that $\beta^t$ increases with $t$.
Using the same process, we can see that $b^0 \leq c_j^0 + \ell^t$, where $c_j^0 = \sum_{i=1}^n a_{ji} \theta_i^T h_i^0$.

We can now return to bounding $\alpha^t$ in \eqref{eqn:sep_cons}
Because $G^0 \subseteq \tilde{G}^t$ for all $t \geq 1$, we know that $h^0 $ is in $\tilde{G}^t$ for all $t \geq 1$. We use the following Lemma to show that $b^* \geq c_j$.
\begin{lemma}
\label{lem:bstar}
    Assume the same as Lemma \ref{lem:z_tight} and let $j$ be a constraint satisfying Lemma \ref{lem:z_tight}.
    Then, we have that $\max_{\theta \in \tilde{C}^t} \sum_{i=1}^n a_{ji} \theta_i^{\top} h_i^* \geq c_j$.
\end{lemma}
\begin{proof}
    Let the set $L$ be defined as 
    \begin{equation*}
        L = \{ [h_1^{\top} \ ...\ h_n^{\top}]^{\top} : \max_{\theta \in \tilde{C}^t} \sum_{i=1}^n a_{ji} \theta_i^{\top} h_i \leq c_j \}.
    \end{equation*}
    To prove the Lemma, we show that $h^*$ is not in the interior of $L$.
    
    First, we show that $L$ is convex.
    To do so, it is sufficient to show that $r(h) = \max_{\theta \in \tilde{C}^t} \sum_{i=1}^n a_{ji} \theta_i^{\top} h_i$ is convex.
    We can show this by noting that $r$ can be expressed as $r(h) = \max_{\phi \in \Phi} \phi^{\top} h$, where $\Phi = \{ [a_{j1} \theta_1^{\top}\ ...\ a_{jn} \theta_n^{\top}]^{\top} : (\theta_1,...,\theta_n) \in \tilde{C}^t \}$.
    Note that $r$ is the support function on $\Phi$, which is necessarily convex \cite{boyd2004convex}.
    Therefore, $L$ is convex.
    
    Since $L$ is convex, so is the interior of $L$ \cite{bertsekas2009convex}.
    By Lemma \ref{lem:z_tight}, $z^t = \alpha^t h^* + (1 - \alpha^t) h^0$ is in the boundary of $L$ and hence not in the interior.
    Also note that $h^0$ is in the interior of $L$ because it is in the interior of $\tilde{G}^t$ (as $\mathbf{int} L \supseteq \mathbf{int}\tilde{G}^t$).
    Suppose that $h^*$ is in the interior of $L$.
    This would imply, by the definition of a convex set, that $z^t$ is in the interior of $L$.
    This is a contradiction, so $h^*$ cannot be in the interior of $L$.
    As this corresponds to the statement of the Lemma, the proof is complete.
\end{proof}

Therefore, we know that $b^* \geq c_j$ and $b^0 \leq c_j$ (as $[h_1^{0,\top}\theta_1\ ...\ h_n^{0,\top}\theta_n]^{\top}$ is in $E$ for all $\theta$ in $C^0$), which imply that $b^* - b^0 \geq 0$ and $\tilde{c}_j = c_j - b^0 \geq 0$.
With these facts in mind and the bounds on $b^*$ and $b^0$, we can continue to simplify the constraint as follows.
\begin{align}
    & c_j \leq \alpha^t b^* + (1 - \alpha^t) b^0 \\
    \Rightarrow \ & \alpha^t \geq \frac{c_j - b^0}{b^* - b^0} = \frac{\tilde{c}_j}{b^* - b^0} \geq \frac{\tilde{c}_j}{c_j + \ell^t - b^0} = \frac{\tilde{c}_j}{\tilde{c}_j + \ell^t}
\end{align}
From the bound on $b^0$, we also have that
\begin{equation}
    \tilde{c}_j = c_j - b^0 \geq c_j - c_j^0 - \ell^t \geq \zeta - \ell^t,
\end{equation}
where the last inequality follows from Assumption \ref{ass:init_set} which implies that for any $h$ in $G^0$ (such as $h^0$), it holds that $c_j - \sum_{i=1}^n a_{ji} h_i^{\top} \theta_i^*  \geq \zeta$ for all $j$ in $[p]$.
This gives us a bound on $1 - \alpha^t$,
\begin{equation}
\label{eqn:alph_bound}
    1 - \alpha^t \leq \frac{\ell^t}{\tilde{c}_j + \ell^t} \leq \frac{\ell^t}{\zeta}.
\end{equation}
Therefore, we have the bound on $r_{t+1}^I$ for $T \geq t+1 > T' \geq t_{\delta}$:
\begin{equation}
    r_{t+1}^I \leq \frac{4 \sqrt{2} \kappa M n^2 L^2 S \sqrt{\beta^T}}{\zeta \sqrt{2 \nu + \lambda_- T'}}
\end{equation}
Note that this bound only holds if both the events in Theorem \ref{thm:conf_set} and Lemma \ref{lem:min_eig} hold.
Since each of these events happen with probability at least $1 - \delta$, they jointly hold with probability at least $1 - 2 \delta$ by the union bound.

\subsubsection{Bounding $r_t^{II}$}
\label{sec:apx_rtII}
In general, our approach for bounding $r_t^{II}$ is standard for stochastic linear bandit analysis (e.g. \cite{abbasi2011improved,dani2008stochastic}).
However, our work does differ in that we use the Lipschitz property of $f$ and only consider $T > T'$.

Also note that the statements in this section only hold if $\theta_i^*$ is in $C_i^t$ for all $i$ in $[n]$ and $t$ in $[T]$ which holds with probability at least $1 - \delta$ by Theorem \ref{thm:conf_set}.
For brevity, this is only discussed here.

Using $r_{t,i}^{II}$ as $r_t^{II}$ due to user $i$, we can use the Lipschitz property in Assumption \ref{ass:lipschitz} and standard linear bandit analysis to get that for $t > T'$:
\begin{align*}
r_{t,i}^{II} & = f_i (h_i(\gamma_i^t)^{\top} \tilde{\theta}_i^t) - f_i (h_i(\gamma_i^t)^{\top} \theta_i^* )\\
& \leq | f_i (h_i(\gamma_i^t)^{\top} \tilde{\theta}_i^t) - f_i (h_i(\gamma_i^t)^{\top} \theta_i^* ) | \\
& \leq M | h_i(\gamma_i^t)^{\top} (\tilde{\theta}_i^t -  \theta_i^* ) | \\
& = M | h_i(\gamma_i^t)^{\top} (\tilde{\theta}_i^t - \hat{\theta}_i^{t-1} + \hat{\theta}_i^{t-1} -  \theta_i^* ) | \\
&\
\begin{aligned}
\leq & \ M \| h_i (\gamma_i^t) \|_{[V_i^{t-1}]^{-1}}\\
& \times ( \| \tilde{\theta}_i^t - \hat{\theta}_i^{t-1} \|_{V_i^{t-1}} + \| \hat{\theta}_i^{t-1} - \theta_i^* \|_{V_i^{t-1}} )
\end{aligned}\\
& = 2 M \| h_i (\gamma_i^t) \|_{[V_i^{t-1}]^{-1}} \sqrt{\beta^{t-1}} \label{ellip} \numberthis
\end{align*}
With the trivial bound $r_{t,i}^{II} \leq 2MLS$ and assuming for simplicity that $T$ is large enough such that $\beta^T \geq 1$, we get  that
\begin{align*}
    r_{t,i}^{II} & \leq \min( 2 M \| h_i (\gamma_i^t) \|_{[V_i^{t-1}]^{-1}} \sqrt{\beta^{t-1}}, 2 MLS)\\
   & \leq 2 M \max(LS,1) \min( \| h_i (\gamma_i^t) \|_{[V_i^{t-1}]^{-1}} \sqrt{\beta^{t-1}}, 1)\\
   & \leq 2 M \max(LS,1) \sqrt{\beta^T} \min( \| h_i (\gamma_i^t) \|_{[V_i^{t-1}]^{-1}} , 1) \numberthis
\end{align*}
We then use a standard linear bandit Lemma to proceed.
\begin{lemma} 
\label{lem:ellip_pot}
(Lemma 11 in \cite{abbasi2011improved}) For the sequence $\{ h_i^t \}_{t=1}^\infty$ in $\mathbb{R}^m$, $\nu > 0$ and $V_i^t = \nu I + \sum_{s=1}^t h_i^s h_i^{s,\top}$ such that $\| h_i^t \| \leq L$ for all $t \in [T]$, we have that
\begin{equation*}
\begin{split}
    & \sum_{t=1}^T \min( \| h_i \|_{[V_i^{t-1}]^{-1}}^2 , 1  ) \\
    & \leq 2(m \log((\textrm{trace}(\nu I) + T L^2)/m) - \log \text{det} (\nu I)).
\end{split}
\end{equation*}
\end{lemma}
Lemma \ref{lem:ellip_pot} is used as follows.
\begin{align*}
    & \sum_{t=T'+1}^T \min( \| h_i (\gamma_i^t) \|_{[V_i^{t-1}]^{-1}}^2 , 1  )\\
    & \leq \sum_{t=1}^T \min( \| h_i^t \|_{[V_i^{t-1}]^{-1}}^2 , 1  )\\
    & \leq 2(m \log((\text{trace}(\nu I) + T L^2)/m) - \log \text{det} (\nu I))\\
    & = 2(m \log((m \nu + T L^2)/m) -  m \log(\nu))\\
    & = 2m \log( 1 + T L^2/ (m \nu) ) \numberthis
\end{align*}
Cauchy-Schwarz can then be used to bound the cumulative regret due to $r_{t,i}^{II}$:
\begin{equation}
    \sum_{t=T'+1}^T r_{t,i}^{II} \leq \sqrt{ (T - T') \sum_{t=T' + 1}^T [r_{t,i}^{II}]^2 }
\end{equation}
Applying the upper bounds on $\sum_{t=T' + 1}^T [r_{t,i}^{II}]^2$ yields the complete bound for the cumulative regret due to $r_{t,i}^{II}$ which is given in the next section.

\subsubsection{Complete regret bound}

With probability at least $ 1- 2 \delta$, it holds that
\begin{align*}
R_T^I & = \sum_{t=1}^{T'} r_t + \sum_{t=T'+1}^T r_t^I + \sum_{t=T'+1}^T r_t^{II} \\
&
\begin{aligned}
\leq\ & 2 MnLST' + \frac{4 \sqrt{2} (T - T') \kappa M n^2 L^2 S \sqrt{\beta^T}}{\zeta \sqrt{2 \nu + \lambda_- T'}}\\
&  + n M \max(LS,1) \sqrt{4 (T - T')} \\
& \ \ \times \sqrt{\beta^T 2m \log( 1 + T L^2/ (m \nu) )}
\end{aligned} \numberthis
\end{align*}
The first term above comes from a trivial bound on $r_t$:
\begin{align*}
    r_t & = \sum_{i=1}^n [ f_i (h_i(\gamma_i^*)^{\top}\theta_i^*) - f_i (h_i(\gamma_i^t)^{\top} \theta_i^*) ]\\  
    & \leq \sum_{i=1}^n | f_i (h_i(\gamma_i^*)^{\top}\theta_i^*) - f_i (h_i(\gamma_i^t)^{\top} \theta_i^*) | \\
    & \leq M \sum_{i=1}^n | (h_i(\gamma_i^*) - h_i(\gamma_i^t))^{\top} \theta_i^* |\\
    & \leq M \sum_{i=1}^n \| h_i(\gamma_i^*) - h_i(\gamma_i^t) \| \| \theta_i^* \| \\
    & \leq 2 M n L S \numberthis
\end{align*}

\subsection{Proof of Theorem \ref{thm:no_reg}}
\label{sec:apx_thm3}

The regret analysis for this setting draws heavily from the PR regret analysis, but it also requires additional work to handle the error in the second argument of the utility function (i.e. the fact that $\munderbar{f}_i(\cdot, \munderbar{\check{\theta}}^t_i)$ does not necessarily equal $\munderbar{f}_i(\cdot, \theta_i^*)$ because $\munderbar{\check{\theta}}^t_i$ is not necessarily equal to $\theta_i^*$).
To deal with this, we use the decomposition
\begin{align*}
\munderbar{r}_t & = \sum_{i=1}^n \left[ \munderbar{f}_i (x_i (\munderbar{\gamma}_i^*, \theta_i^*),\theta_i^*) - \munderbar{f}_i (x_i (\munderbar{\gamma}_i^t,\theta_i^*),\theta_i^*) \right]\\
& = \munderbar{r}_t^I + \munderbar{r}_t^{II} + \munderbar{r}_t^{III} + \munderbar{r}_t^{IV} \numberthis \label{eqn:no_decomp}
\end{align*}
where
\begin{align*}
    \munderbar{r}_t^I & = \sum_{i=1}^n [ \munderbar{f}_i(h_i(\munderbar{\gamma}_i^*)^{\top} \theta_i^*, \theta_i^*) -  \munderbar{f}_i(h_i(\munderbar{\gamma}_i^*)^{\top} \theta_i^*, \munderbar{\check{\theta}}_i^t)]\\
    \munderbar{r}_t^{II} & = \sum_{i=1}^n [ \munderbar{f}_i(h_i(\munderbar{\gamma}_i^*)^{\top} \theta_i^*, \munderbar{\check{\theta}}_i^t) -  \munderbar{f}_i(h_i(\munderbar{\gamma}_i^t)^{\top} \munderbar{\tilde{\theta}}_i^t, \munderbar{\check{\theta}}_i^t)]\\
    \munderbar{r}_t^{III} & = \sum_{i=1}^n [ \munderbar{f}_i(h_i(\munderbar{\gamma}_i^t)^{\top} \munderbar{\tilde{\theta}}_i^t, \munderbar{\check{\theta}}_i^t) -  \munderbar{f}_i(h_i(\munderbar{\gamma}_i^t)^{\top} \munderbar{\tilde{\theta}}_i^t, \theta_i^*)]\\
    \munderbar{r}_t^{IV} & = \sum_{i=1}^n [ \munderbar{f}_i(h_i(\munderbar{\gamma}_i^t)^{\top} \munderbar{\tilde{\theta}}_i^t, \theta_i^*) -  \munderbar{f}_i(h_i(\munderbar{\gamma}_i^t)^{\top} \theta_i^*, \theta_i^*)].
\end{align*}
At first glance, we can see that terms $\munderbar{r}_t^I$ and $\munderbar{r}_t^{III}$ track the variation in the second argument of the utility function, while terms $\munderbar{r}_t^{II}$ and $\munderbar{r}_t^{IV}$ track the variation in the first argument of the utility function.
In order to bound both sets of terms, we use a Lipschitz property on the utility function as we did in the PR analysis.
However, the analysis here differs from the PR analysis because we need that the utility function is Lipschitz with respect to both arguments and because the utility function is defined with respect to the price response (rather than defined arbitrarily) in this setting.

We use the decomposition described in \eqref{eqn:no_decomp} and bound each of the terms separately.
There is a section for preliminary analysis followed by sections for bounding each of the terms in \eqref{eqn:no_decomp}.

\subsubsection{Preliminary analysis}
\label{sec:apx_no_pre}

In this section, we show that $\munderbar{f}_i(x_i, \theta_i)$ can be expressed in closed-form, is increasing and that it is Lipschitz with respect to both $x_i$ and $\theta_i$.
First, we prove that the partial derivative of $\munderbar{f}_i(x_i,\theta_i)$ with respect to $x_i$ must be equal to the inverse of the price response.
This property allows $\munderbar{f}_i$ to be expressed in closed-form and is also used to show that $\munderbar{f}_i(x_i,\theta_i)$ is increasing with respect to $x_i$.
Throughout this section, we denote the inverse of the price response $g_i (x_i, \theta_i)$, such that $x_i( g_i (\tau , \theta_i), \theta_i) = \tau $ for all $\tau$ in $\mathbb{R}_{++}$ and $g_i( x_i ( y , \theta_i), \theta_i) = y $ for all $y$ in $\mathbb{R}_{++}$.
\begin{lemma}
\label{lem:no_equiv}
    Fix some nonzero $\theta_i$ in $\mathbb{R}^m$.
    Then, $\frac{\partial}{\partial x_i} \munderbar{f}_i (x_i, \theta_i) = g_i (x_i, \theta_i)$, where $g_i ( \cdot , \cdot)$ is the inverse of the price response $x_i (\cdot,\cdot)$ with respect to the first argument such that $x_i( g_i (\tau , \theta_i), \theta_i) = \tau $ for all $\tau$ in $\mathbb{R}_{++}$ and $g_i( x_i ( y , \theta_i), \theta_i) = y $ for all $y$ in $\mathbb{R}_{++}$.
\end{lemma}
\begin{proof}
Note that $\left( \munderbar{f}_i(x_i, \theta_i) - \gamma_i x_i \right)$ is differentiable with respect to $x_i$ and the interval $\mathbb{R}_{++}$ is open.
Therefore, we have from Fermat's Theorem that the maximizer of $\left( \munderbar{f}_i(x_i, \theta_i) - \gamma_i x_i \right)$, which we denote $x_i^*$, necessarily satisfies
\begin{equation*}
    \frac{\partial}{\partial x_i} [\munderbar{f}_i(x_i, \theta_i) - \gamma_i x_i] \Big \rvert_{x_i = x_i^*} = 0.
\end{equation*}
Evaluating the derivative yields
\begin{equation}
\label{eqn:inv_rel}
    \frac{\partial}{\partial x_i} \munderbar{f}_i(x_i^*, \theta_i) = \gamma_i.
\end{equation}
It follows that $\frac{\partial}{\partial x_i} \munderbar{f}_i(x_i(\gamma_i, \theta_i), \theta_i) = \gamma_i$ for all $\gamma_i \in \mathbb{R}_{++}$.
For clarity, we use $F_i(x_i) = \frac{\partial}{\partial x_i} \munderbar{f}_i(x_i, \theta_i) $ and $x_i(\gamma_i) = x_i(\gamma_i, \theta_i)$ for the remainder of the proof.
Note that $x_i(\gamma_i)$ is strictly monotone and hence invertible because $h_i$ is strictly decreasing and $\theta_i$ is fixed as nonzero.
From \eqref{eqn:inv_rel}, we have that
\begin{equation*}
    F_i(x_i(\gamma_i)) = \gamma_i.
\end{equation*}
for all $\gamma_i \in \mathbb{R}_{++}$.
This states that $F_i$ is the left-inverse of $x_i(\cdot)$.
Since $x_i(\cdot)$ is invertible, this left-inverse is unique and hence the inverse of $x_i(\cdot)$ (\cite{keisler1996mathematical} Prop. A.7.3).
Therefore, we have that $g_i(x_i,\theta_i) = \frac{\partial}{\partial x_i} \munderbar{f}_i(x_i, \theta_i)$.
\end{proof}

It immediately follows from Lemma \ref{lem:no_equiv} that the utility can be expressed as
\begin{equation}
\label{eqn:no_util}
    \munderbar{f}_i(x_i, \theta_i) = \munderbar{f}_i(x_i^0, \theta_i) + \int_{x^0_i}^{x_i} g_i (\tau , \theta_i) d\tau
\end{equation}
for nonzero $\theta_i$.
This provides a closed-form equation for the utility function that can be used by the algorithm to compute the utility for a given $x_i$ and $\theta_i$.
We also use Lemma \ref{lem:no_equiv} to establish that $\munderbar{f}_i$ must be increasing in the following lemma.

\begin{lemma}
Let Assumption \ref{ass:no_resp} hold. Then, $\munderbar{f}_i(\cdot, \theta_i)$ is strictly increasing such that for all $x_i^1 < x_i^2$ we have that $\munderbar{f}_i(x_i^1, \theta_i) < \munderbar{f}_i(x_i^2, \theta_i)$ for any nonzero $\theta_i$. \label{lem:no_incr}
\end{lemma}
\begin{proof}
We use the Mean Value Theorem to show that $\munderbar{f}_i$ is strictly increasing with respect to $x_i$.
The Mean Value Theorem is applicable because $\munderbar{f}_i$ is differentiable with respect to $x_i$.
From the Mean Value Theorem, we have that for any $x_i^1$ and $x_i^2$, there exists $x_i^3 \in (x_i^1,x_i^2)$ such that
\begin{equation*}
     \munderbar{f}_i(x_i^1, \theta_i) - \munderbar{f}_i(x_i^2, \theta_i) =  (x_i^1 - x_i^2) \frac{\partial}{\partial x_i}  \munderbar{f}_i (x_i^3, \theta_i)
\end{equation*}
Therefore, if $\frac{\partial \munderbar{f}_i (x_i, \theta_i)}{\partial x_i} > 0$ for all $x_i$, then $\munderbar{f}_i$ is strictly increasing.
Equivalently, we need that
\begin{equation*}
    0 < \frac{\partial}{\partial x_i} \munderbar{f}_i (x_i, \theta_i) = g_i (x_i, \theta_i).
\end{equation*}
This is true if the codomain of $g_i$ is a subset of the positive reals.
Therefore, it is true if the domain of $x_i (\cdot)$ is a subset of the positive reals.
This holds because Assumption \ref{ass:no_resp} ensures that the domain of $h_i$ is the positive reals and that $\theta_i$ is taken to be nonzero.
\end{proof}

Lemma \ref{lem:no_incr} is necessary in order to apply the analysis from the PR setting.
Also note that it is expected that the utility function is increasing given that one's benefit derived from a resource would generally increase as the resource consumption increases.

We then prove that $\munderbar{f}_i$ is Lipschitz with respect to both arguments.

\begin{lemma}
\label{lem:lipsch_x}
Let Assumption \labelcref{ass:no_bound} hold and fix some $\theta$ in $\munderbar{C}^0$. For all $x$ in $E$, the function $\munderbar{f}_i(x_i,\theta_i)$ is $\Gamma$-Lipschitz continuous with respect to $x_i$ for all $i$ in $[n]$.
That is, for any $x^1, x^2$ in $ E$,
\begin{equation*}
    | \munderbar{f}_i(x_i^1, \theta_i) - \munderbar{f}_i(x_i^2, \theta_i) | \leq \Gamma | x_i^1 - x_i^2 |,
\end{equation*}
for all $i$ in $[n]$.
\end{lemma}

\begin{proof}
With the Mean Value Theorem, we have that for any $x_i^1, x_i^2$, there exists a $x_i^3 \in (x_i^1, x_i^2)$ such that
\begin{equation*}
    \munderbar{f}_i(x_i^1, \theta_i) - \munderbar{f}_i(x_i^2, \theta_i) = (x_i^1 - x_i^2) \frac{\partial}{\partial x_i} \munderbar{f}_i (x_i^3, \theta_i)
\end{equation*}
It is easy to see that
\begin{equation*}
    | \munderbar{f}_i(x_i^1, \theta_i) - \munderbar{f}_i(x_i^2, \theta_i) | \leq  |x_i^1 - x_i^2| \sup_{x_i^3 \in (x_i^1, x_i^2)} \left| \frac{\partial}{\partial x_i}  \munderbar{f}_i (x_i^3, \theta_i) \right|
\end{equation*}
Therefore, we can upper bound the absolute value of the derivative to derive a Lipschitz constant.
From Lemma \ref{lem:no_equiv}, the derivative of $\munderbar{f}_i(x_i, \theta_i)$ with respect to $x_i$ is
\begin{equation*}
    \frac{\partial}{\partial x_i} \munderbar{f}_i (x_i^3, \theta_i) = g_i (x_i^3, \theta_i).
\end{equation*}
We can then apply Assumption \ref{ass:no_bound} to bound $g_i (x_i^3, \theta_i)$ such that
\begin{equation*}
    \left| \frac{\partial}{\partial x_i} \munderbar{f}_i (x_i^3, \theta_i) \right| = | g_i (x_i^3, \theta_i) | \leq \Gamma,
\end{equation*}
which completes the proof.
\end{proof}

\begin{lemma}
\label{lem:lipsch_th}
Let Assumptions \labelcref{ass:no_resp,ass:no_bound} hold and fix some $x$ in $E$. For all $\theta$ in $\munderbar{C}^0$, the function $\munderbar{f}_i(x_i,\theta_i)$ is $(\eta + \frac{\xi L}{\rho K})$-Lipschitz continuous with respect to $\theta_i$ for all $i$ in $[n]$.
That is, for any $\theta^1, \theta^2$ in $\munderbar{C}^0$,
\begin{equation*}
    | \munderbar{f}_i(x_i, \theta_i^1) - \munderbar{f}_i(x_i, \theta_i^2) | \leq \left( \eta + \frac{\xi L}{\rho K}\right) \| \theta_i^1 - \theta_i^2 \|,
\end{equation*}
for all $i$ in $[n]$.
\end{lemma}
\begin{proof}
Let $f_i^0(\cdot) = \munderbar{f}_i(x_i^0, \cdot)$.
First, we manipulate the left-hand side of the inequality in the lemma:
\begin{align*}
    & | \munderbar{f}_i(x_i, \theta_i^1) - \munderbar{f}_i(x_i, \theta_i^2) |\\
    &
    \begin{aligned}
    & = \Bigg| f_i^0(\theta_i^1) + \int_{x^0_i}^{x_i} g_i (\tau_i , \theta_i^1) d\tau_i \\
    & \qquad - \left( f_i^0(\theta_i^2) + \int_{x^0_i}^{x_i} g_i (\tau_i , \theta_i^2) d\tau_i \right) \Bigg|
    \end{aligned}\\
    &
    \begin{aligned}
    & \leq \left| f_i^0(\theta_i^1) - f_i^0(\theta_i^2) \right| \\
    & \qquad + \left| \int_{x^0_i}^{x_i} g_i (\tau_i , \theta_i^1) d\tau_i - \int_{x^0_i}^{x_i} g_i (\tau_i , \theta_i^2) d\tau_i \right|
    \end{aligned}
\end{align*}
Next, we consider the gradient of $G(x_i,\theta_i) = \int_{x^0_i}^{x_i} g_i (\tau_i , \theta_i^1) d\tau_i$ with respect to $\theta_i$ for some $\theta$ in $\munderbar{C}^0$,
\begin{equation*}
    \nabla_{\theta_i} \int_{x^0_i}^{x_i} g_i (\tau_i, \theta_i) d \tau_i = \int_{x^0_i}^{x_i} \nabla_{\theta_i}  g_i (\tau_i, \theta_i) d \tau_i,
\end{equation*}
which follows from the Leibniz integral rule. 
In order to evaluate $\nabla_{\theta_i}   g_i (\tau_i , \theta_i)$, we use the fact that $g_i (\tau_i , \theta_i)$ is defined implicitly with the relation
\begin{equation*}
    R(g_i (\tau_i , \theta_i), \theta_i) = \theta_i^T h_i (g_i (\tau_i , \theta_i)) - \tau_i = 0,
\end{equation*}
which follows from the fact that $g_i$ is the inverse of $x_i$.
Therefore, for some $\gamma_i = g_i (\tau_i , \theta_i)$, we can use implicit differentiation to get
\begin{equation}
\label{eqn:imp_diff}
    \nabla_{\theta_i} g_i (\tau_i , \theta_i) = \frac{\nabla_{\theta_i} R(\gamma_i, \theta_i)}{\frac{\partial}{\partial \gamma_i} R(\gamma_i, \theta_i)} = \frac{h_i (\gamma_i)}{\theta_i^T h_i'(\gamma_i)}.
\end{equation}
Since $[\tau_1\ ...\ \tau_n]^{\top}$ is in $\bar{E}$ by convexity and $\theta$ is in $\munderbar{C}^0$ by specification, we know that $h_i'(\gamma_i)$ is nonzero due to Assumption \ref{ass:no_bound}.
Also, the function $h_i$ is differentiable by definition.
Therefore, we have that the numerator in \eqref{eqn:imp_diff} exists for all $\gamma_i$ and the denominator is nonzero (and exists) for all $\theta_i$ in $\munderbar{C}^0$ and all $\gamma_i$.
Therefore, $G(x_i,\theta_i)$ is differentiable with respect to $\theta_i$ for all $\theta_i$ in $\munderbar{C}^0$ and all $\gamma_i$.

This allows us to apply the Mean Value Theorem, which gives that
\begin{equation*}
    | G_i(x_i, \theta_i^1) - G_i(x_i, \theta_i^2) | \leq \sup_{\theta_i^3 \in \munderbar{C}^0 } \| \nabla_{\theta_i} G_i (x_i, \theta_i^3) \|  \| \theta_i^1 - \theta_i^2 \|.
\end{equation*}
In order to bound $\| \nabla_{\theta_i} G_i (x_i, \theta_i) \|$, we first bound $ \|\nabla_{\theta_i} g_i (\tau_i , \theta_i) \|$ as
\begin{equation*}
    \|\nabla_{\theta_i} g_i (\tau_i , \theta_i) \| = \frac{ \| h_i (g_i (\tau_i , \theta_i)) \| }{ | \theta_i^T h_i'(g_i (\tau_i , \theta_i)) |} \leq \frac{L}{\rho K}.
\end{equation*}
This follows from Assumptions \labelcref{ass:no_resp,ass:no_bound}, because $\tau=[\tau_1\ ...\ \tau_n]^{\top}$ is in $\bar{E}$ and $\theta$ is in $\munderbar{C}^0$.
The reason that $\tau$ is in $\bar{E}$ is because $x^0$ and $x$ are in $\bar{E}$, and $\bar{E}$ is convex.
Therefore, we have that
\begin{align*}
    \| \nabla_{\theta_i} G_i (x_i, \theta_i) \| & \leq \left \| \int_{x^0_i}^{x_i} \nabla_{\theta_i}  g_i (\tau_i, \theta_i) d \tau_i \right \| \\
    & \leq \int_{x^0_i}^{x_i} \left \| \nabla_{\theta_i}  g_i (\tau_i, \theta_i) \right \| d \tau_i  \\
    & \leq | x_i - x^0_i |  \max_{\tau_i \in [x^0_i, x_i]} \| \nabla_{\theta_i}  g_i (\tau_i , \theta_i) \| \\
    & \leq \frac{\xi L}{\rho K}.
\end{align*}

We can also bound $\left| f_i^0(\theta_i^1) - f_i^0(\theta_i^2) \right|$ with Assumption \ref{ass:no_bound} as
\begin{equation*}
    \left| f_i^0(\theta_i^1) - f_i^0(\theta_i^2) \right| \leq  \eta \left \| \theta_i^1 - \theta_i^2 \right \|
\end{equation*}

This allows a Lipschitz constant to be derived:
\begin{align*}
    & | \munderbar{f}_i(x_i, \theta_i^1) - \munderbar{f}_i(x_i, \theta_i^2) |\\
    &
    \begin{aligned}
    & \leq \left| f_i^0(\theta_i^1) - f_i^0(\theta_i^2) \right| \\
    & \qquad + \left| G_i (x_i, \theta_i^1) - G_i (x_i, \theta_i^2) \right|
    \end{aligned}\\
    & \leq \left( \eta + \frac{\xi L}{\rho K} \right) \left \| \theta_i^1 - \theta_i^2 \right \|
\end{align*}
This completes the proof.
\end{proof}

Lemmas \labelcref{lem:lipsch_x,lem:lipsch_th} allow the regret to be put in terms of the price response for terms $\munderbar{r}_t^{II}$ and $\munderbar{r}_t^{IV}$, and in terms of the parameter for terms $\munderbar{r}_t^I$ and $\munderbar{r}_t^{III}$.

\subsubsection{Bounding $\underline{r}_t^I$ and $\underline{r}_t^{III}$}
We can first use Lemma \ref{lem:lipsch_th} to bound $\munderbar{r}_t^I$ with the error in the parameter as
\begin{align*}
    \munderbar{r}_t^I & = \sum_{i=1}^n [ \munderbar{f}_i(h_i(\munderbar{\gamma}_i^*)^{\top} \theta_i^*, \theta_i^*) -  \munderbar{f}_i(h_i(\munderbar{\gamma}_i^*)^{\top} \theta_i^*, \munderbar{\check{\theta}}_i^t)]\\
    & \leq \sum_{i=1}^n | \munderbar{f}_i(h_i(\munderbar{\gamma}_i^*)^{\top} \theta_i^*, \theta_i^*) -  \munderbar{f}_i(h_i(\munderbar{\gamma}_i^*)^{\top} \theta_i^*, \munderbar{\check{\theta}}_i^t)|\\
    & \leq \left( \eta + \frac{\xi L}{\rho K} \right) \sum_{i=1}^n \| \theta_i^* -  \munderbar{\check{\theta}}_i^t \|, \numberthis
\end{align*}
where Lemma \ref{lem:lipsch_th} applies because $h_i(\munderbar{\gamma}_i^*)^{\top} \theta_i^*$ is in $\bar{E}$ by definition, and both $\theta_i^*$ and $\munderbar{\check{\theta}}_i^t$ are in $\munderbar{C}^0$.
Note that $\munderbar{\check{\theta}}_i^t$ is in $\munderbar{C}^0$ because $\munderbar{\check{\theta}}_i^t$ is in $\munderbar{C}_i^t$ by definition and $\munderbar{C}_i^t \subseteq \munderbar{C}^0$.

We can bound $\munderbar{r}_t^{III}$ similarily:
\begin{align*}
    \munderbar{r}_t^{III} & = \sum_{i=1}^n [ \munderbar{f}_i(h_i(\munderbar{\gamma}_i^t)^{\top} \munderbar{\tilde{\theta}}_i^t, \munderbar{\check{\theta}}_i^t) -  \munderbar{f}_i(h_i(\munderbar{\gamma}_i^t)^{\top} \munderbar{\tilde{\theta}}_i^t, \theta_i^*)]\\
    & \leq \sum_{i=1}^n | \munderbar{f}_i(h_i(\munderbar{\gamma}_i^t)^{\top} \munderbar{\tilde{\theta}}_i^t, \munderbar{\check{\theta}}_i^t) -  \munderbar{f}_i(h_i(\munderbar{\gamma}_i^t)^{\top} \munderbar{\tilde{\theta}}_i^t, \theta_i^*)|\\
    & \leq \left( \eta + \frac{\xi L}{\rho K} \right) \sum_{i=1}^n \| \theta_i^* -  \munderbar{\check{\theta}}_i^t \|, \numberthis
\end{align*}
where we can apply Lemma \ref{lem:lipsch_th} because $h_i(\munderbar{\gamma}_i^t)^{\top} \munderbar{\tilde{\theta}}_i^t$ is guaranteed to be in $E$ by definition.

To proceed, we can define an expanded confidence set as we did in the PR setting (in \eqref{eqn:exp_conf}):
\begin{equation}
\label{eqn:no_exp_conf}
\begin{split}
    \munderbar{\tilde{C}}_i^t = \{ & \theta_i \in \mathbb{R}^m_+ : \\
    & \| \theta_i - \theta_i^* \|_{V_i^t} \leq 2 \sqrt{\beta^t}, \| \theta_i \| \leq S, \mathbf{1}^{\top} \theta_i \geq \rho \}
\end{split}
\end{equation}
As before, we know that $\munderbar{C}_i^t \subseteq \munderbar{\tilde{C}}_i^t$.
Using Lemma \ref{lem:min_eig}, we can then see that for $t \geq T'$, any $\theta_i$ in $\munderbar{\tilde{C}}_i^t$ satisfies
\begin{align*}
    2 \sqrt{\beta^t} & \geq \| \theta_i - \theta_i^* \|_{V_i^t}\\
    & \geq \| \theta_i - \theta_i^* \| \sqrt{ \lambda_{\text{min}} (V_i^t)}\\
    & \geq \| \theta_i - \theta_i^* \| \sqrt{\lambda_{\text{min}} (V_i^{T'})}\\
    & \geq \| \theta_i - \theta_i^* \| \sqrt{\nu + \frac{\lambda_- T'}{2}} \numberthis
\end{align*}
Since $\munderbar{\check{\theta}}_i^t$ is in $\munderbar{\tilde{C}}_i^{t-1}$ and all terms are positive, we then have that
\begin{equation}
    \| \theta_i^* -  \munderbar{\check{\theta}}_i^t \| \leq \frac{2 \sqrt{2 \beta^{t-1}}}{\sqrt{2\nu + \lambda_- T'}} \leq \frac{2 \sqrt{2 \beta^T}}{\sqrt{2\nu + \lambda_- T'}},
\end{equation}
where we use that $\beta^T \geq \beta^t$ for $T \geq t$.
Therefore, it follows that
\begin{align}
    \munderbar{r}_t^I & \leq \frac{2 \sqrt{2} n (\eta + \frac{\xi L}{\rho K}) \sqrt{\beta^T}}{ \sqrt{2\nu + \lambda_- T'}},\\
    \munderbar{r}_t^{III} & \leq \frac{2 \sqrt{2} n (\eta + \frac{\xi L}{\rho K}) \sqrt{\beta^T}}{\sqrt{2\nu + \lambda_- T'}}.
\end{align}

\subsubsection{Bounding $\underline{r}_t^{II}$}
We use a similar process to bound $r_t^I$ in the PR setting (Appendix \ref{sec:apx_rt1}), but need to redefine some of the sets.
First, recall the definition of $\munderbar{\tilde{C}}_i^t$ in \eqref{eqn:no_exp_conf}.
Using the same process (i.e. by including the restriction that $\mathbf{1}^{\top} \theta_i^* \leq \rho$ from Assumption \ref{ass:no_resp}), we then define $\munderbar{\tilde{D}}^t$ similar to \eqref{eqn:shr_pset}, $\munderbar{\tilde{G}}^t$ similar to \eqref{eqn:lconf_set} and $\munderbar{G}^0$ similar to \eqref{eqn:g_init}.
Taking $\munderbar{h}^* = [h_1(\munderbar{\gamma}_1^*)^{\top}\ ...\ h_n(\munderbar{\gamma}_n^*)^{\top}]^{\top}$ and $\munderbar{h}^0$ be any element in $\munderbar{G}^0$, we have that
\begin{equation}
    \munderbar{\alpha}^t = \max \{ \munderbar{\alpha} \in [0,1] : \munderbar{\alpha} \munderbar{h}^* + (1 - \munderbar{\alpha}) \munderbar{h}^0 \in \munderbar{\tilde{G}}^t \}.
\end{equation}
Let $\munderbar{z}^t = \munderbar{\alpha}^t \munderbar{h}^* + (1 - \munderbar{\alpha}^t) \munderbar{h}^0$ and $\munderbar{z}_i^t = \munderbar{\alpha}^t \munderbar{h}_i^* + (1 - \munderbar{\alpha}^t) \munderbar{h}_i^0$.
Using Lemma \ref{lem:lipsch_x} and the fact that $\sum_{i=1}^n f_i (h_i(\munderbar{\gamma}_i^{t+1})^{\top} \munderbar{\tilde{\theta}}_i^{t+1}, \munderbar{\check{\theta}}_i^{t+1}) \geq \sum_{i=1}^n f_i ([\munderbar{z}_i^t]^{\top}\theta_i^*, \munderbar{\check{\theta}}_i^{t+1})$ (given that the left hand side is chosen optimistically), we have that
\begin{align*}
    \munderbar{r}_{t+1}^{II} & = \sum_{i=1}^n [ \munderbar{f}_i(h_i(\munderbar{\gamma}_i^*)^{\top} \theta_i^*, \munderbar{\check{\theta}}_i^{t+1}) -  \munderbar{f}_i(h_i(\munderbar{\gamma}_i^{t+1})^{\top} \munderbar{\tilde{\theta}}_i^{t+1}, \munderbar{\check{\theta}}_i^{t+1})]\\
    & \leq \sum_{i=1}^n [ \munderbar{f}_i(h_i(\munderbar{\gamma}_i^*)^{\top} \theta_i^*, \munderbar{\check{\theta}}_i^{t+1}) -  \munderbar{f}_i ([\munderbar{z}_i^t]^{\top}\theta_i^*, \munderbar{\check{\theta}}_i^{t+1})]\\
    & \leq \sum_{i=1}^n | \munderbar{f}_i(h_i(\munderbar{\gamma}_i^*)^{\top} \theta_i^*, \munderbar{\check{\theta}}_i^{t+1}) -  \munderbar{f}_i([\munderbar{z}_i^t]^{\top}\theta_i^*, \munderbar{\check{\theta}}_i^{t+1})|\\
    & \leq \Gamma \sum_{i=1}^n | h_i(\munderbar{\gamma}_i^*)^{\top} \theta_i^* -  [\munderbar{z}_i^t]^{\top}\theta_i^*|\\
    & = \Gamma \sum_{i=1}^n | (\munderbar{h}_i^* - \munderbar{h}_i^0) ^{\top} \theta_i^* | (1 - \munderbar{\alpha}^t) \\
    & \leq \Gamma \sum_{i=1}^n  \| \munderbar{h}_i^* - \munderbar{h}_i^0 \| \| \theta_i^* \| (1 - \munderbar{\alpha}^t) \\
    & \leq 2 \Gamma n L S (1 - \munderbar{\alpha}^t) \numberthis
\end{align*}
Note that Lemma \ref{lem:lipsch_x} can be applied above because $h_i(\munderbar{\gamma}_i^*)^{\top} \theta_i^*$ and $[\munderbar{z}_i^t]^{\top}\theta_i^*$ are in $E$ and 
$\check{\munderbar{\theta}}_i^{t+1}$ is in
$\munderbar{C}_0$.

We then proceed similarly to the PR setting in that we show that the constraint is tight on the optimal solution.
To do so, we state several lemmas which are equivalent to Lemmas \labelcref{lem:equiv_sol,lem:tight_const,lem:z_tight,lem:bstar} from the PR setting.

\begin{lemma}
    (Equivalent to Lemma \ref{lem:equiv_sol})
    Let Assumption \labelcref{ass:no_resp} hold. Then the optimal consumption $\munderbar{x}^* = [h_1(\munderbar{\gamma}_1^*)^{\top}\theta_1^*\ ...\ h_n(\munderbar{\gamma}_n^*)^{\top} \theta_n^*]^{\top}$ satisfies
    \begin{equation*}
    \munderbar{x}^* \in \argmax_{x \in \munderbar{E}} \sum_{i=1}^n \munderbar{f}_i (x_i,\theta_i^*),
    \end{equation*}
    where
    \begin{equation*}
    \munderbar{E} = \left\{ x\in \mathbb{R}^n_{++} :  \sum_{i=1}^n a_{ji} x_i \leq c_j,\ \forall j \in [p] \right\}.
    \end{equation*}
\label{lem:no_equiv_sol}
\end{lemma}

\begin{lemma}
\label{lem:no_tight_const}
(Equivalent to Lemma \ref{lem:tight_const})
Let Assumption \labelcref{ass:no_resp} hold. Then, there exists a constraint $j$ in $[p]$ such that $\sum_{i=1}^n a_{ji} \munderbar{x}_i^* = c_j$, where $\munderbar{x}^*$ is an optimal consumption as defined in Lemma \ref{lem:no_equiv_sol}.
\end{lemma}

\begin{lemma}
\label{lem:no_z_tight}
(Equivalent to Lemma \ref{lem:z_tight})
    Let Assumption \labelcref{ass:no_resp} hold.
    Then with probability at least $1 - \delta$, there exists $j$ in $[p]$ such that
    \begin{equation*}
        \max_{\theta \in \munderbar{\tilde{C}}^t} \sum_{i=1}^n a_{ji} \theta_i^{\top} \munderbar{z}_i^t = c_j.
    \end{equation*}
\end{lemma}

\begin{lemma}
\label{lem:no_bstar}
    (Equivalent to Lemma \ref{lem:bstar})
    Assume the same as Lemma \ref{lem:no_z_tight} and let $j$ be a constraint satisfying Lemma \ref{lem:no_z_tight}.
    Then, we have that $\max_{\theta \in \munderbar{\tilde{C}}^t} \sum_{i=1}^n a_{ji} \theta_i^{\top} \munderbar{h}_i^* \geq c_j$.
\end{lemma}

The above lemmas hold because of the similarity between the two settings and because Assumption \ref{ass:no_resp} is stronger than Assumption \ref{ass:price_resp}, which is used in the PR counterparts of the above lemmas.

We can then use the condition in Lemma \ref{lem:no_z_tight} to bound $\munderbar{\alpha}^t$ similar to \eqref{eqn:sep_cons}.
\begin{align*}
    c_j & = \max_{\theta \in \munderbar{\tilde{C}}^t} \sum_{i=1}^n a_{ji} [\munderbar{z}_i^t]^{\top}\theta_i\\
    & \leq \munderbar{\alpha}^t \underbrace{\sum_{i=1}^n \max_{\theta_i \in \munderbar{\tilde{C}}^t_i} a_{ji} [\munderbar{h}_i^*]^{\top} \theta_i}_{\munderbar{b}^*} + (1 - \munderbar{\alpha}^t) \underbrace{\sum_{i=1}^n \max_{\theta_i \in \munderbar{\tilde{C}}^t_i} a_{ji} [\munderbar{h}_i^0]^{\top} \theta_i}_{\munderbar{b}^0} \numberthis
\end{align*}

Note that the bound on $\munderbar{b}^{*}$ and $\munderbar{b}^0$ in \eqref{eqn:bstar} applies.
Accordingly, we have that $\munderbar{b}^* \leq c_j + \ell^t$ and $\munderbar{b}^0 \leq c_j^0 + \ell^t$, where $\ell^t$ is defined in \eqref{eqn:bstar} and $\munderbar{c}_j^0 = \sum_{i=1}^n a_{ji} \theta_i^T \munderbar{h}_i^0$.
Then, it follows from Lemma \ref{lem:no_bstar} that $\munderbar{b}^* \geq c_j$, $\munderbar{b}^0 \leq c_j$.
Using the PR analysis in Appendix \ref{sec:apx_rt1}, we can see that that the bound in \eqref{eqn:alph_bound} applies such that $1 - \munderbar{\alpha}^t \leq \ell^t / \zeta$.
Therefore, the bound on $r_t^{II}$ for $T \geq t+1 > T'$ is
\begin{equation}
    r_{t+1}^{II} \leq \frac{4 \sqrt{2} \kappa \Gamma n^2 L^2 S \sqrt{\beta^T}}{\zeta \sqrt{2 \nu + \lambda_- T'}},
\end{equation}
which holds with probability at least $1 - 2 \delta$.

\subsubsection{Bounding $\underline{r}_t^{IV}$}
The analysis for term $\underline{r}_t^{IV}$ is very similar to that for term $r_t^{II}$ in Appendix \ref{sec:apx_rtII} which itself is similar to the stochastic linear bandit analysis.

With $\munderbar{r}^{IV}_{t,i}$ as $\munderbar{r}^{IV}_t$ due to user $i$, we can use Lemma \ref{lem:lipsch_x} for $t \geq T'$ as
\begin{align*}
    \munderbar{r}^{IV}_{t,i} & = \munderbar{f}_i(h_i(\munderbar{\gamma}_i^t)^{\top} \munderbar{\tilde{\theta}}_i^t, \theta_i^*) -  \munderbar{f}_i(h_i(\munderbar{\gamma}_i^t)^{\top} \theta_i^*, \theta_i^*)\\
    & \leq | \munderbar{f}_i(h_i(\munderbar{\gamma}_i^t)^{\top} \munderbar{\tilde{\theta}}_i^t, \theta_i^*) -  \munderbar{f}_i(h_i(\munderbar{\gamma}_i^t)^{\top} \theta_i^*, \theta_i^*) |\\
    & \leq \Gamma | h_i(\munderbar{\gamma}_i^t)^{\top} (\munderbar{\tilde{\theta}}_i^t - \theta_i^*)|\\
    & = \Gamma | h_i(\munderbar{\gamma}_i^t)^{\top} (\munderbar{\tilde{\theta}}_i^t - \hat{\theta}_i^{t-1} + \hat{\theta}_i^{t-1} -  \theta_i^*)|\\
    &\
    \begin{aligned}
    \leq & \ \Gamma \| h_i (\munderbar{\gamma}_i^t) \|_{[V_i^{t-1}]^{-1}}\\
    & \times ( \| \tilde{\theta}_i^t - \hat{\theta}_i^{t-1} \|_{V_i^{t-1}} + \| \hat{\theta}_i^{t-1} - \theta_i^* \|_{V_i^{t-1}} )
    \end{aligned}\\
    & = 2 \Gamma \| h_i (\munderbar{\gamma}_i^t) \|_{[V_i^{t-1}]^{-1}} \sqrt{\beta^{t-1}}, \numberthis
\end{align*}
where the last inequality holds with probability at least $1 - \delta$.
We then use the trivial bound $\munderbar{r}^{IV}_{t,i} \leq 2 \Gamma L S$ to get
\begin{align*}
    \munderbar{r}^{IV}_{t,i} & \leq \min( 2 \Gamma \| h_i (\munderbar{\gamma}_i^t) \|_{[V_i^{t-1}]^{-1}} \sqrt{\beta^{t-1}}, 2 \Gamma LS)\\
   & \leq 2 \Gamma \max(LS,1) \min( \| h_i (\munderbar{\gamma}_i^t) \|_{[V_i^{t-1}]^{-1}} \sqrt{\beta^{t-1}}, 1)\\
   & \leq 2 \Gamma \max(LS,1) \sqrt{\beta^T} \min( \| h_i (\munderbar{\gamma}_i^t) \|_{[V_i^{t-1}]^{-1}} , 1), \numberthis
\end{align*}
where the last inequality assumes (for simplicity) that $T$ is large enough such that $\beta^T \geq 1$.

It then follows from Lemma \ref{lem:ellip_pot} that 
\begin{align*}
    & \sum_{t=T'+1}^T \min( \| h_i (\munderbar{\gamma}_i^t) \|_{[V_i^{t-1}]^{-1}}^2 , 1  )\\
    & \leq \sum_{t=1}^T \min( \| h_i (\munderbar{\gamma}_i^t) \|_{[V_i^{t-1}]^{-1}}^2 , 1  )\\
    & \leq 2(m \log((\text{trace}(\nu I) + T L^2)/m) - \log \text{det} (\nu I))\\
    & = 2(m \log((m \nu + T L^2)/m) -  m \log(\nu))\\
    & = 2m \log( 1 + T L^2/ (m \nu) ) \numberthis
\end{align*}
Using Cauchy-Schwarz as
\begin{equation}
    \sum_{t=T'+1}^T \munderbar{r}_{t,i}^{IV} \leq \sqrt{ (T - T') \sum_{t=T' + 1}^T [\munderbar{r}_{t,i}^{IV}]^2 },
\end{equation}
we can apply the bounds on $\sum_{t=T' + 1}^T [\munderbar{r}_{t,i}^{IV}]^2$ to get the bound on $\munderbar{r}_{t,i}^{IV}$.

\subsection{Complete regret bound}
Combining all of the instantaneous regret terms, we have with probability at least $1 - 2 \delta$ that
\begin{align*}
    \munderbar{R}_T^I & =  \sum_{t=1}^{T'} \munderbar{r}_t + \sum_{t=T'+1}^T \munderbar{r}_t^I + \sum_{t=T'+1}^T \munderbar{r}_t^{II} + \sum_{t=T'+1}^T \munderbar{r}_t^{III} \\
    & \quad + \sum_{t=T'+1}^T \munderbar{r}_t^{IV}\\
    & \leq 2 \Gamma nLST'  + \frac{2 \sqrt{2} n (T - T') (\eta + \frac{\xi L}{\rho K}) \sqrt{\beta^T}}{ \sqrt{2\nu + \lambda_- T'}}\\
    & \quad + \frac{4 \sqrt{2} (T - T') \kappa \Gamma n^2 L^2 S \sqrt{\beta^T}}{\zeta \sqrt{2 \nu + \lambda_- T'}}\\
    & \quad + n \Gamma \max(LS,1) \sqrt{8 (T - T') \beta^T m \log \left( 1 + \frac{T L^2}{m \nu} \right)} \numberthis
\end{align*}
Note that the trivial bound on $\munderbar{r}_t$ is derived as follows.
\begin{align*}
    \munderbar{r}_t & = \sum_{i=1}^n \left[ \munderbar{f}_i (h_i(\munderbar{\gamma}_i^*)^{\top} \theta_i^*,\theta_i^*) - \munderbar{f}_i (h_i(\munderbar{\gamma}_i^t)^{\top} \theta_i^*,\theta_i^*) \right]\\
    & \leq \sum_{i=1}^n \left| \munderbar{f}_i (h_i(\munderbar{\gamma}_i^*)^{\top} \theta_i^*,\theta_i^*) - \munderbar{f}_i (h_i(\munderbar{\gamma}_i^t)^{\top} \theta_i^*,\theta_i^*) \right|\\
    & \leq \Gamma \sum_{i=1}^n \left| h_i(\munderbar{\gamma}_i^*)^{\top} \theta_i^* - h_i(\munderbar{\gamma}_i^t)^{\top} \theta_i^* \right|\\
    & = \Gamma \sum_{i=1}^n \left| (h_i(\munderbar{\gamma}_i^*) - h_i(\munderbar{\gamma}_i^t))^{\top} \theta_i^* \right|\\
    & \leq \Gamma \sum_{i=1}^n \| h_i(\munderbar{\gamma}_i^*) - h_i(\munderbar{\gamma}_i^t) \| \| \theta_i^* \| \\
    & \leq 2 \Gamma n L S \numberthis
\end{align*}

\subsection{The SUM Utility Function in the  DR Setting}
\label{sec:apx_sumdr}

In this section, we describe how the SUM utility function $\munderbar{f}_i$ is defined and calculated in the DR problem in Section \ref{sec:dr}.
Since the consumption is vector-valued in this case, the formulation of the SUM utility function needs to be extended.
In the DR problem, the price response $x_i(\gamma_i,\theta_i)$ is vector-valued (with dimension $V$) to represent the consumption at each time period in the day.
Also, the price for each user $\gamma_i$ is a vector (of dimension $V$) in the DR problem.
Since the agent is assumed to be profit-maximizing, the SUM utility function $\munderbar{f}_i : \mathbb{R}_{++}^V \times \mathbb{R}^m \rightarrow \mathbb{R} $ is implicitly defined as
\begin{equation}
    x_i(\gamma_i, \theta_i) = \argmax_{x_i \in \mathbb{R}_{++}^V} \left( \munderbar{f}_i(x_i, \theta_i) - \gamma_i^\top x_i \right).
\end{equation}
Therefore, under the argument in Lemma \ref{lem:no_equiv} (and relevant assumptions), we have that
\begin{equation}
    \nabla_{x_i} \munderbar{f}_i(x_i, \theta_i) = g_i(x_i, \theta_i),
\end{equation}
where $g_i$ is the inverse of $x_i(\cdot)$ with respect to the first argument as discussed in Appendix \ref{sec:apx_no_pre}.
Therefore (as in \eqref{eqn:no_util}), $\munderbar{f}_i$ can be written as
\begin{equation}
\label{eqn:dr_sum_util}
    \munderbar{f}_i (x_i, \theta_i) = \munderbar{f}_i (x_i^0, \theta_i) + \int_\xi  g_i(\tau, \theta_i)^{\top} d \tau,
\end{equation}
by the gradient theorem, where $\xi$ is a continuous curve which starts at $x_i^0$ and ends at $x_i$.

However, using \eqref{eqn:dr_sum_util} to calculate the SUM utility directly is challenging because the price response model presented in Section \ref{sec:dr} is highly irregular, making it challenging to efficiently evaluate the inverse price response.
In order to handle this, we use the smoothed price response function,
\begin{equation}
    \tilde{x}_i (\gamma_i, \theta_i) = \frac{1}{2 V} \sum_{v=1}^V \big( x_i (\gamma_i + \rho \mathbf{e}_v, \theta_i) + x_i (\gamma_i - \rho \mathbf{e}_v, \theta_i) \big),
\end{equation}
where $\mathbf{e}_v$ is a $V$ dimensional vector with a 1 at position $v$ and $0$ everywhere else.
We denote the inverse of $\Tilde{x}_i$ with respect to the first argument as $\Tilde{g}_i$.
We approximate $\Tilde{g}_i(x_i, \theta_i)$ at some $x_i$ by finding the $\gamma_i$ such that $\Tilde{x}_i(\gamma_i, \theta_i) - x_i = 0$ with a root-finding solver.

We also assume that $\munderbar{f}_i (x_i^0, \theta_i) = 0$ for any $\theta_i$, which will not change the result of the optimistic update \eqref{eqn:no_ofu_price} in the SUM algorithm since a constant addition will not change the result of an argmax.
To approximate the integral in \eqref{eqn:dr_sum_util}, we use a midpoint Riemann sum such that the approximate SUM utility is
\begin{equation}
\label{eqn:sum_util_approx}
    \tilde{\munderbar{f}}_i (x_i, \theta_i) = \sum_{k=0}^{K-1} \tilde{g}_i \left( x_0 + \left(k + 0.5 \right) \Delta , \theta_i \right)^{\top} \Delta
\end{equation}
where $\Delta = \frac{1}{K}(x_i - x_0)$.
In the simulations, we approximate the SUM utility with \eqref{eqn:sum_util_approx}, we choose $x_0 = \mathbf{1}$, $K = 5$, and $\rho = 1$.

\end{document}